\documentclass[10pt,conference,letterpaper]{IEEEtran}
\IEEEoverridecommandlockouts

\usepackage{cite}
\usepackage{amsmath,amssymb,amsfonts}
\usepackage{graphicx}
\usepackage{textcomp}
\usepackage{xcolor}
\usepackage{booktabs}
\usepackage{dblfloatfix}
\usepackage{cancel}

\usepackage{algorithm}
\usepackage{algpseudocode}

\usepackage[inline]{enumitem}
\usepackage{subcaption}
\usepackage{siunitx}

\usepackage{pf2}
\usepackage{ntheorem}

\usepackage{tikz}
\usepackage{pgfplotstable}
\pgfplotsset{compat=1.16}
\usepgfplotslibrary{fillbetween}

\usetikzlibrary{patterns, patterns.meta} 
\usetikzlibrary{calc}
\usetikzlibrary{positioning}
\usetikzlibrary{arrows.meta}
\usetikzlibrary{math} 
\usetikzlibrary{backgrounds,scopes}
\usetikzlibrary{pgfplots.statistics, pgfplots.colorbrewer} 
\usetikzlibrary{trees} 
\usetikzlibrary {datavisualization.formats.functions}
\usetikzlibrary{fit}
\usetikzlibrary{shapes.geometric}
\usetikzlibrary{shadows}
\usetikzlibrary{shapes}

\newtheorem{theorem}{Theorem}
\newtheorem{corollary}{Corollary}
\newtheorem{lemma}{Lemma}
\newtheorem{definition}{Definition}

\newcommand{\dproc}[1]{\ensuremath{d_{proc}(#1)}}
\newcommand{\dprop}[1]{\ensuremath{d_{prop}(#1)}}

\newcommand{\drate}[1]{\ensuremath{d_{rate}(#1)}}

\newcommand{\dpdb}[1]{\ensuremath{\mathbf{d}\ifthenelse{\equal{#1}{}}{}{(#1)}}}
\newcommand{\dpdbmin}[1]{\ensuremath{\mathbf{d}^{min}\ifthenelse{\equal{#1}{}}{}{(#1)}}}
\newcommand{\dpdbmax}[1]{\ensuremath{\mathbf{d}^{max}\ifthenelse{\equal{#1}{}}{}{(#1)}}}
\newcommand{\Dpdb}[1]{\ensuremath{\mathbf{d}\ifthenelse{\equal{#1}{}}{}{\big(#1\big)}}}
\newcommand{\Dpdbmin}[1]{\ensuremath{\mathbf{d}^{min}\ifthenelse{\equal{#1}{}}{}{\big(#1\big)}}}
\newcommand{\Dpdbmax}[1]{\ensuremath{\mathbf{d}^{max}\ifthenelse{\equal{#1}{}}{}{\big(#1\big)}}}
\newcommand{\dpdbtransmin}[1]{\ensuremath{\mathbf{d}_{trans}^{min}\ifthenelse{\equal{#1}{}}{}{\big(#1\big)}}}
\newcommand{\dpdbtransmax}[1]{\ensuremath{\mathbf{d}_{trans}^{max}\ifthenelse{\equal{#1}{}}{}{\big(#1\big)}}}
\newcommand{\dpdbfull}[1]{\ensuremath{\big[\dpdbmin{#1}, \dpdbmax{#1}\big]}}

\newcommand{\TT}{\ensuremath{\mathcal{F}}}

\newcommand{\path}[1]{\ensuremath{#1.\textit{path}}}
\newcommand{\size}[1]{\ensuremath{#1.\textit{size}}}
\newcommand{\period}[1]{\ensuremath{#1.\textit{period}}}
\newcommand{\release}[1]{\ensuremath{#1.\textit{release}}}
\newcommand{\phase}[1]{\ensuremath{#1.\textit{phase}}}
\newcommand{\ete}[1]{\ensuremath{#1.\textit{lat}}}
\newcommand{\jitter}[1]{\ensuremath{#1.\textit{jitter}}}
\newcommand{\rel}[1]{\ensuremath{#1.rel}}

\renewcommand{\v}[2]{\ensuremath{v_{\scriptscriptstyle#2}^{\scriptscriptstyle #1}}}
\newcommand{\vl}[1]{\ensuremath{v_{\scriptscriptstyle#1}^{\scriptscriptstyle n(#1)}}}

\newcommand{\Conf}{\ensuremath{\mathcal{C}}}
\renewcommand{\S}[1]{\ensuremath{\mathcal{S}\ifthenelse{\equal{#1}{}}{}{\big(#1\big)}}}
\newcommand{\Smin}[1]{\ensuremath{\mathcal{S}^{min}\ifthenelse{\equal{#1}{}}{}{\big(#1\big)}}}
\newcommand{\Smax}[1]{\ensuremath{\mathcal{S}^{max}\ifthenelse{\equal{#1}{}}{}{\big(#1\big)}}}
\newcommand{\R}[1]{\ensuremath{\mathcal{R}\ifthenelse{\equal{#1}{}}{}{\big(#1\big)}}}
\newcommand{\Rmin}[1]{\ensuremath{\mathcal{R}^{min}\ifthenelse{\equal{#1}{}}{}{\big(#1\big)}}}
\newcommand{\Rmax}[1]{\ensuremath{\mathcal{R}^{max}\ifthenelse{\equal{#1}{}}{}{\big(#1\big)}}}

\newcommand{\GCL}{\ensuremath{\mathcal{S}_{GCL}}}
\newcommand{\SE}[1]{\ensuremath{\mathcal{S}^\mathcal{E}\ifthenelse{\equal{#1}{}}{}{\big(#1\big)}}}
\newcommand{\IE}[1]{\ensuremath{B^\mathcal{E}\ifthenelse{\equal{#1}{}}{}{\big(#1\big)}}}
\newcommand{\Si}[2]{\ensuremath{\mathcal{S}_{#2}\ifthenelse{\equal{#1}{}}{}{\big(#1\big)}}}

\newcommand{\I}[1]{\ensuremath{B\ifthenelse{\equal{#1}{}}{}{\big(#1\big)}}}

\newcommand{\add}[3]{\ensuremath{\big(#1{} + #2{}\big)\big(#3\big)}}
\newcommand{\sub}[3]{\ensuremath{\big(#1{} - #2{}\big)\big(#3\big)}}

\newcommand{\Prob}{\ensuremath{\mathbb{P}}}

\newcommand{\E}{\ensuremath{\mathcal{E}}}
\newcommand{\T}[1]{\ensuremath{\mathcal{T}\ifthenelse{\equal{#1}{}}{}{\big(#1\big)}}}
\newcommand{\D}[1]{\ensuremath{\mathcal{D}\ifthenelse{\equal{#1}{}}{}{\big(#1\big)}}}

\newcommand{\C}{\ensuremath{\mathcal{C}}}
\newcommand{\V}{\ensuremath{\mathcal{V}}}

\def\displayexplan{1}
\newcommand{\estepref}[2]{$\langle#1\rangle#2$}

\newcommand{\eD}[1]{\ensuremath{\mathcal{D}^e\ifthenelse{\equal{#1}{}}{}{\big(#1\big)}}}
\newcommand{\eDtrans}[1]{\ensuremath{\mathcal{D}_{trans}^e\ifthenelse{\equal{#1}{}}{}{\big(#1\big)}}}
\newcommand{\Dtrans}[1]{\ensuremath{\mathcal{D}_{trans}\ifthenelse{\equal{#1}{}}{}{\big(#1\big)}}}

\newcommand{\IT}{\textsf{\small IsochronousTalker}}
\newcommand{\TC}{\textsf{\small TransmissionConsistency}}
\newcommand{\ST}{\textsf{\small SequentialTransmission}}
\newcommand{\FIFO}{\textsf{\small FIFO}}
\newcommand{\GCLC}{\textsf{\small GCL-Encapsulation}}
\newcommand{\GCLS}{\textsf{\small GCL-Progress}}
\newcommand{\PC}{\textsf{\small PolicingConsistency}}
\newcommand{\PSFP}{\textsf{\small PSFP}}
\newcommand{\TP}{\textsf{\small TransmissionPolicing}}

\newcommand{\pclink}[2]{\ensuremath{[\v{#1-1}{#2}, \v{#1}{#2}]}}
\newcommand{\clink}[2]{\ensuremath{[\v{#1}{#2}, \v{#1+1}{#2}]}}
\newcommand{\nclink}[2]{\ensuremath{[\v{#1+1}{#2}, \v{#1+2}{#2}]}}

\newcommand{\thinskip}[0]{\hspace*{0.16667em}\relax}
\newcommand{\jrdash}[2]{\unskip #1\thinskip #2\thinskip\ignorespaces}

\newcommand{\ldash}[0]{\jrdash{\empty}{\hbox{---}\nobreak}}
\newcommand{\rdash}[0]{\jrdash{\nobreak}{---}}

\def\showtodos{1}
\newcommand{\todo}[1]{\ifthenelse{\equal{\showtodos}{1}}{\textcolor{lred}{(TODO: #1)}}{\ignorespaces}}

\def\err at (#1,#2){
    \coordinate (a) at (#1,#2);
    \draw[black!80,ultra thick] ($(a) + 1/3*(-0.5,1.5)$) -- ($(a) + 1/3*(0.5,0.5)$);
    \draw[black!80,ultra thick] ($(a) + 1/3*(0.5,1.5)$) -- ($(a) + 1/3*(-0.5,0.5)$);
}

\def\intervalplt at (#1,#2,#3){
    \draw (axis cs: #1, #2) edge[very thick] (axis cs: #1, #3);
    \draw (axis cs: #1-0.2, #2) edge[very thick] (axis cs: #1+0.2, #2);
    \draw (axis cs: #1-0.2, #3) edge[very thick] (axis cs: #1+0.2, #3);
}

\pgfmathdeclarefunction{fpumod}{2}{%
    \pgfmathfloatdivide{#1}{#2}%
    \pgfmathfloatint{\pgfmathresult}%
    \pgfmathfloatmultiply{\pgfmathresult}{#2}%
    \pgfmathfloatsubtract{#1}{\pgfmathresult}%
    \pgfmathfloatifapproxequalrel{\pgfmathresult}{#2}{\def\pgfmathresult{5}}{}%
}

\pgfdeclarepatternformonly{queue}{\pgfqpoint{-5pt}{-5pt}}{\pgfqpoint{5pt}{5pt}}{\pgfqpoint{2pt}{6.15pt}}%
{
  \pgfsetlinewidth{0.4pt}
  \pgfpathmoveto{\pgfqpoint{0pt}{0pt}}
  \pgfpathlineto{\pgfqpoint{3.1pt}{0pt}}
  \pgfpathmoveto{\pgfqpoint{0pt}{0pt}}
  \pgfusepath{stroke}
}

\definecolor{F1}{gray}{0.2}
\definecolor{F2}{gray}{0.4}
\definecolor{F3}{gray}{0.6}

\definecolor{f1}{gray}{0.4}
\definecolor{f2}{gray}{0.6}
\definecolor{f3}{gray}{0.8}

\definecolor{green}{HTML}{D5E8D4}
\definecolor{blue}{HTML}{DAE8FC}
\definecolor{test}{HTML}{D8E8E8}
\definecolor{orange}{HTML}{FFE6CC}
\definecolor{red}{HTML}{F8CECC}

\definecolor{lgreen}{HTML}{82B366}
\definecolor{lblue}{HTML}{6C8EBF}
\definecolor{lorange}{HTML}{D79B00}
\definecolor{lred}{HTML}{B85450}

\def\BibTeX{{\rm B\kern-.05em{\sc i\kern-.025em b}\kern-.08em
    T\kern-.1667em\lower.7ex\hbox{E}\kern-.125emX}}
\begin{document}

\title{End-to-End Reliability in Wireless IEEE~802.1Qbv Time-Sensitive Networks}

\author{\IEEEauthorblockN{Simon Egger$^\ast$, James Gross$^\dagger$, Joachim Sachs$^\mathsection$, Gourav P. Sharma$^\dagger$, Christian Becker$^\ast$, Frank Dürr$^\ast$}
\IEEEauthorblockA{{}$^\ast$University of Stuttgart, {}$^\dagger$KTH Royal Institute of Technology, Stockholm, {}$^\mathsection$Ericsson Research, Aachen\\
\{simon.egger,$\,$christian.becker,$\,$frank.duerr\}@ipvs.uni-stuttgart.de, \{jamesgr,$\,$gpsharma\}@kth.se, joachim.sachs@ericsson.com
}}

\maketitle

\begin{abstract}
  Industrial cyber-physical systems require dependable network communication with formal end-to-end reliability guarantees.
  Striving towards this goal, recent efforts aim to advance the integration of 5G into Time-Sensitive Networking~(TSN).
  However, we show that IEEE~802.1Qbv TSN schedulers that are unattuned to 5G packet delay variations may jeopardize any reliability guarantees provided by the 5G system.
  We demonstrate this on a case where a $\qty{99.99}{\percent}$ reliability in the inner 5G network diminishes to below $\qty{10}{\percent}$ when looking at end-to-end communication in TSN.
  In this paper, we overcome this shortcoming by introducing Full Interleaving Packet Scheduling~(FIPS) as a wireless-friendly IEEE~802.1Qbv scheduler.
  To the best of our knowledge, FIPS is the first to provide formal end-to-end QoS guarantees in wireless TSN.
  FIPS allows a controlled batching of TSN streams, which improves schedulability in terms of the number of wireless TSN streams by a factor of up to~$\times 45$.
  Even in failure cases, FIPS isolates the otherwise cascading QoS violations to the affected streams and protects all other streams.  
  With formal end-to-end reliability, improved schedulability, and fault isolation, FIPS makes a substantial advance towards dependability in wireless TSN. 
\end{abstract}

\section{Introduction}
Time-Sensitive Networking (TSN) and 5G are widely recognized as key enabling technologies for dependable (wireless) industrial networking~\cite{9299391}.
Envisioned use cases \ldash like networked control systems to govern a fleet of automated guided vehicles~(AGV)~\cite{3gpp.22.104, det6g_usecases} \rdash require ultra-reliable low latency communication, making end-to-end reliability and end-to-end latency the key performance indicators for dependability.
3GPP defines the reliability of a stream as the percentage of frames that arrive within their time constraints~\cite{3gpp.22.261}. 
However, it is crucial to note that 5G reliability solely pertains to the inner mobile networks and does not provide end-to-end reliability in TSN.
Neglecting this difference can nullify any quality-of-service~(QoS) guarantee made by the 5G system.
For instance, our simulation results show a $\qty{99.99}{\percent}$ 5G reliability plummeting to an end-to-end reliability of below $\qty{10}{\percent}$.

Time-Sensitive Networking captures a set of standards under IEEE 802.1, ranging from a general TSN bridge specification~\cite{802.1Q} to specific profiles for industrial automation~\cite{60802}.
In this work, we focus on scheduling time-triggered streams in networks with both wired and wireless elements.
Time-triggered traffic is periodic and specifies its QoS requirements through bounded end-to-end reliability, end-to-end latency, and arrival time jitter.
Compared to analytical techniques that inspect the guarantees of a given schedule (e.g., network calculus~\cite{le2001network}) or frameworks that provide network diagnostics (e.g., fault localization~\cite{10682902}), we consider the problem of synthesizing a schedule that satisfies given QoS requirements.
For this purpose, we employ the IEEE 802.1Qbv Time-Aware Shaper~(TAS)~\cite{802.1Qbv} to forward frames according to precise timetables, called gate control lists (GCLs), along each hop.
Computing GCLs that satisfy the QoS requirements of all streams is, in general, an NP-hard optimization problem.

An extensive body of research exists on synthesizing GCLs in wired TSN, e.g., with ILP solvers~\cite{10.1145/3139258.3139289}, SMT solvers~\cite{8894249}, (meta-)heuristics~\cite{nwps}, or deep learning~\cite{10228875}.
However, TAS is highly susceptible to runtime uncertainties, with timing inaccuracies in the range of mere microseconds or sporadic frame loss causing total schedule breakdowns~\cite{Craciunas2016RTNS}.
In comparison, 5G introduces stochastic packet delays with variations that are often in the range of milliseconds~\cite{downlink_example_histogram}.
This makes conventional scheduling techniques with deterministic system models unsuitable (e.g., the above~\cite{nwps,10228875,10.1145/3139258.3139289,8894249}), as they would have to rely on (non-robust) scalar or worst-case 5G channel assumptions.
In alignment with state-of-the-art solutions of related scheduling domains~\cite{infocom24_best_paper,diaz2023robust}, we therefore advocate for moving towards stochastic and robust scheduling approaches that can provide formal end-to-end QoS guarantees.

Within this context, we introduce Full Interleaving Packet Scheduling~(FIPS) as a wireless-friendly IEEE 802.1Qbv scheduler that does not rely on non-robust 5G channel assumptions.
Instead, FIPS computes 5G packet delay budgets (PDBs) based on the streams QoS requirements and the 5G link information, e.g., through packet delay histograms.
PDBs specify the target delay for the 5G system and are of the form:
``With a 5G reliability of $\qty{99.99}{\percent}$, the 5G packet delays for stream $F$ are lower- and upper-bounded by the budget interval $[\qty{3}{\ms}, \qty{15}{\ms}]$.''
Under the condition that the 5G system satisfies this requirement, FIPS extends the QoS guarantees end-to-end, for example, to:
``With an end-to-end reliability of $\qty{99.99}{\percent}$, each frame of $F$ arrives at the TSN listener with a latency below $\qty{20}{\ms}$ and jitter below $\qty{100}{\us}$.''

We summarize our main contributions as follows:

\textbf{Formal End-to-End QoS:}
FIPS computes TSN configurations that are robust against bounded runtime uncertainties, yielding (provable) end-to-end QoS guarantees as above.

\textbf{Fault Isolation:}
Errors do not cascade through the entire network but stay isolated to the affected streams.
In particular, FIPS continuously upholds the QoS of wired streams even if the 5G system can no longer sustain the 5G PDBs.

\textbf{Improved Schedulability:} 
FIPS relaxes conventional temporal isolation constraints from IEEE 802.1Qbv scheduling techniques in wired networks~\cite{nwps,Craciunas2016RTNS}.
Our evaluations show that this relaxation improves schedulability in terms of the number of wireless streams by a factor of up to $\times 45$.

The remainder of this paper is structured as follows: 
Section~\ref{sec:related_work} discusses related work.
Section~\ref{sec:background} provides background on 5G/TSN and defines our system model.
Section~\ref{sec:problem_description} captures the problems of current IEEE 802.1Qbv scheduling techniques in wireless TSN and summarizes our problem statement.
Sections~\ref{sec:robustness_and_feasibility} and Section~\ref{sec:fips} introduces the concept of robust scheduling and defines FIPS.
Finally, Section~\ref{sec:eval} contains our evaluation results and Section~\ref{sec:conclusion} concludes this work.

\section{Related Work} \label{sec:related_work}
Scheduling in IEEE 802.1Qbv is closely related to other NP-hard optimization problems.
We therefore review literature both within and outside of the TSN scheduling domain.

\subsection{Scheduling in Time-Sensitive Networks}
System components in wired Time-Sensitive Networks are often modelled to be strictly deterministic to configure the IEEE 802.1Qbv GCLs~\cite{nwps,10228875,10.1145/3139258.3139289,8894249}.
Approaches that consider uncertain runtime behavior most commonly focus on time synchronization errors~\cite{10228980,10144237,Craciunas2016RTNS}, sporadic frame loss~\cite{Craciunas2016RTNS,10.1145/3458768,FENG2022102381,9488750}, or permanent link failures~\cite{9488750,8607244,8474201,reusch2022dependability}.
While these solutions remain applicable in wired network partitions, none of them are designed to cope with non-negligible 5G packet delay variations in the range of milliseconds.

Out of the above solutions, \cite{nwps} and~\cite{Craciunas2016RTNS} are most closely related to our approach and considered as state-of-the-art in wired TSN.
The authors strictly separate frame transmissions to isolate potential transmission faults and can account for bounded clock skew.
While their temporal isolation constraints can be extended to cover 5G packet delay variations, our measurements show that these extensions do not scale in terms of the number of accepted wireless streams.
FIPS improves schedulability by allowing a controlled batching of streams.

The approaches in \cite{9212049} and \cite{9940254} aim to advance the integration of 5G and TSN by allowing the IEEE 802.1Qbv scheduler to allocate 5G resource blocks for wireless streams.
However, they rely on scalar (e.g., average, median, or maximum) 5G packet delays, for which we demonstrate that they cannot provide any end-to-end QoS guarantees. 
Moreover, we argue that 5G-internal resource management (e.g., resource allocation~\cite{8736403}, link adaptation~\cite{9488790}, and QoS provisioning in 5G vRAN~\cite{infocom24_5gvran}) remains a complementary task by itself.
Compared to requiring the 5G system to expose its internal state, we therefore design FIPS to work with 5G PDBs.

More recently, the authors of~\cite{10682834} recognized the problem of cascading faults under non-robust TSN configurations.
They propose a shielding mechanism that discards frames with unexpected delays if they endanger the latency guarantees of other streams.
Complementary to their approach, FIPS computes TSN configurations with an increased tolerance for delay variations (i.e., as captured by the 5G PDBs), preventing excessive frame loss caused by such shielding mechanisms.

\subsection{Scheduling with Uncertainty in Related Domains}
Outside of the TSN literature, state-of-the-art scheduling solutions in multi-hop networks provide remarkable near-optimal latency and reliability guarantees~\cite{6155625,infocom24_best_paper}.
However, they often assume high-cost forwarding buffers, e.g., one dedicated buffer per stream per bridge~\cite{6155625}.
This makes them unsuitable for scheduling in TSN, where each egress port is equipped with a maximum of eight FIFO queues.

Job-shop scheduling covers a vast research domain that encounters similar obstacles when transitioning from deterministic to uncertain task durations;
a survey is provided in~\cite{XIONG2022105731}.
For instance, D{\'\i}az et al.~\cite{diaz2023robust} showed that unattuned (e.g., scalar) approaches are not robust against runtime uncertainty.
In alignment with their observations, FIPS is designed to compute TSN configurations that are robust against bounded uncertainty intervals, as captured by the 5G PDBs.

\begin{figure}[b]
  \centering
  \resizebox{0.74\columnwidth}{!}{%
  \begin{tikzpicture}
    \node[draw, trapezium, align=center, minimum width=3.5cm, minimum height=0.55cm, trapezium stretches body] (psfp) at (0,0) {PSFP $\R{}$};
    \draw (0,0.6) edge[->,thick] node[pos=0,yshift=2mm] {ingress port} (psfp);

    \node[align=center,black!70] (cnc) at (3.5,1) {configured by the CNC};
    \draw[black!50] (cnc) edge[->,dashed] (0.8,0);
    \draw[black!50] (cnc) edge[->,dashed] (4,0.2);

    \node[draw,pattern=queue,minimum width=0.6cm,minimum height=1.5cm,rounded corners=2pt] (q0) at (-1.45,-1.5) {};
    \node[draw,pattern=queue,minimum width=0.6cm,minimum height=1.5cm,rounded corners=2pt] (q1) at (-0.55,-1.5) {};
    \node (dots) at (0.45,-1.5) {\Large $\cdots$};
    \node[below=0mm of dots] {\scriptsize $\leq 8$ queues};
    \node[draw,pattern=queue,minimum width=0.6cm,minimum height=1.5cm,rounded corners=2pt] (q2) at (1.45,-1.5) {};

    \draw (q0) edge[<-,thick] node[right,pos=0.6] {0} (\tikztostart |- psfp.south);
    \draw (q1) edge[<-,thick] node[right,pos=0.6] {1} (\tikztostart |- psfp.south);
    \draw (q2) edge[<-,thick] node[right,pos=0.6] {7} (\tikztostart |- psfp.south);

    \node[draw,below=2mm of q0,rounded corners,minimum width=0.6cm,fill=black!60] (g0) {c};
    \node[draw,below=2mm of q1,rounded corners,minimum width=0.6cm,fill=black!60] (g1) {c};
    \node[draw,below=2mm of q2,rounded corners,minimum width=0.6cm,fill=white] (g2) {o};
    \node (egress) at (0,-3.7) {egress port};

    \draw (q0) edge[thick] (g0);
    \draw (g0) edge[->,thick,out=-90,in=90] (egress);
    \draw (q1) edge[thick] (g1);
    \draw (g1) edge[->,thick,out=-90,in=90] (egress);
    \draw (q2) edge[thick] (g2);
    \draw (g2) edge[->,thick,out=-90,in=90] (egress);

    \node (gcl0) at (4,0) {TAS $\GCL$};
    \node (gcl1) at (4,-1.6) {
      \begin{tabular}{ll}
	\toprule
	time & gates\\
	\midrule
	$[0,t_0)$ & oo$\ldots$o\\
	$[t_0,t_1)$ & cc$\ldots$o\\
	$[t_1,t_2)$ & cc$\ldots$o\\
	$[t_2,H)$ & co$\ldots$c\\
      \end{tabular}
    };
    \node[fill=black!20,inner ysep=1pt,inner xsep=2pt] (gcl2) at (4,-2.18) {
      \begin{tabular}{cc}
	$[t_1,t_2)$ & cc$\ldots$o\\
      \end{tabular}
    };
    \node[draw,black!60,rounded corners,fit={(gcl0) (gcl1)}, inner sep=1pt] {};

    \draw (gcl2) edge[dashed,thick,black!50,out=180,in=0] (g2);
    \draw (g2) edge[dashed,thick,black!50] (g1);
    \draw (g1) edge[dashed,thick,black!50] (g0);

    \node[below=0.6cm of gcl2,xshift=7mm,align=left] {o = open\\c = closed};
  \end{tikzpicture}
  }
  \caption{Port-to-port model of TSN bridges with Per-Stream Filtering and Policing~(PSFP) and the Time-Aware Shaper~(TAS).} \label{fig:tsn_model}
\end{figure}
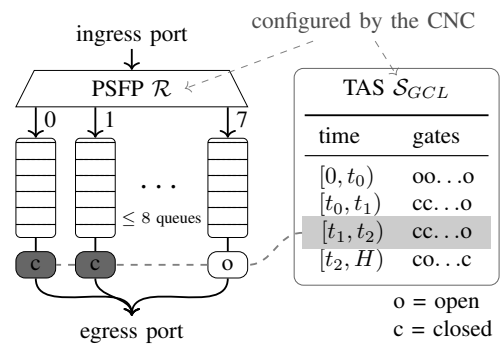

\section{Background and System Model} \label{sec:background}
Next, we provide an overview of relevant standards for Time-Sensitive Networking~(TSN) and its proposed integration with 5G.
We then introduce our system model and define the end-to-end QoS requirements for streams in wireless TSN.

\subsection{Time-Sensitive Networking} \label{sec:tsn}
TSN is a set of IEEE 802.1 standards to provide real-time communication in Ethernet networks with bounded latency and low packet delay variations for each stream.
Streams originate at exactly one talker, traverse multiple bridges, and terminate at one or more listeners.
To schedule time-triggered streams, we consider a centralized network controller~(CNC)~\cite{8514112} with a global view of the network and streams to configure the TSN mechanisms (as described next) at each bridge.
Moreover, we assume a bounded clock skew between network devices~\cite{802.1AS} and that talkers are synchronized to the network schedule.

For time-triggered streams, i.e., periodic traffic with fixed frame sizes, two TSN mechanisms are of central interest: the Time-Aware Shaper~(TAS)~\cite{802.1Qbv} and Per-Stream Filtering and Policing~(PSFP)~\cite{802.1Qci}.
Fig.~\ref{fig:tsn_model} illustrates the operation of TAS and PSFP for a bridge.
When a frame $f$ arrives at the bridge $u$, each frame is first mapped to its corresponding TSN stream via its source/destination MAC address, VLAN identifier, and priority code point~(PCP) value.
PSFP then verifies that the frame adheres to the stream specification, i.e., by arriving within an expected interval $\R{u, f} = [\text{rx}^{min}, \text{rx}^{max}]$. 
If a violation is detected, the frame is discarded; otherwise the frame is enqueued.
To determine the queue where $f$ is appended, the bridge consults its forwarding table (specifying the egress port) and the 3-bit PCP value in the frame's header (specifying one out of eight FIFO queues per egress port).

When the egress port is free, TAS uses gates (one per queue) to select the next frame for transmission.
Whether the gates are currently open or closed is determined by the active entry in the gate control list~(GCL).
While finding suitable GCLs is a computationally hard problem that is done by the CNC in advance, they simplify the transmission selection at runtime:
TAS selects the frame at the head of the highest-priority queue that is non-empty and has an open gate.
In the case of Fig.~\ref{fig:tsn_model}, the third GCL entry is currently active and only opens the gate of the highest-priority queue; 
this allows the first frame of the queue to be forwarded without any delays from lower-priority traffic.
Finally, note that the GCL is also periodic and repeats after the least common multiple of the streams' periods; we call this the hypercycle $H$ of the TSN configuration.

\begin{figure}[b]
  \centering
  \resizebox{\columnwidth}{!}{%
  \begin{tikzpicture}
      \node[draw,minimum width=5.75cm,minimum height=3.5cm,rounded corners] at (2.5,0.25) {};
      \node[draw,minimum width=5cm,fill=white] (5G) at (2.5,2) {Logical 5G-TSN Bridge};

      \node[draw,fill=black!20,minimum height=0.8cm,] (dstt1) at (0,1) {DS-TT};
      \node[draw,minimum height=0.8cm] (ue1) at (0.97,1) {UE};
      \draw (-1,1) edge[thick] (dstt1);

      \node[draw,fill=black!20,minimum height=0.8cm] (dstt2) at (0,-0.9) {DS-TT};
      \node[draw,minimum height=0.8cm] (ue2) at (0.97,-0.9) {UE};
      \draw (-1,-0.9) edge[thick] (dstt2);

      \node[draw,minimum height=0.8cm] (gnb) at (2.5,0.05) {RAN};
      \node[draw,minimum height=0.8cm] (upf) at (4,0.05) {CN};
      \node[draw,fill=black!20,minimum height=0.8cm] (nwtt) at (5,0.05) {NW-TT};

      \draw (ue1) edge[thick,dashed,out=0,in=140] (gnb);
      \draw (ue2) edge[thick,dashed,out=0,in=-140] (gnb);
      \draw (gnb) edge[thick] (upf);
      \draw (nwtt) edge[thick] (6,0.05);

      \node[draw, fill=black!15, shape=cloud, aspect=1.7, inner sep=1pt, align=center] (TSN1) at (7.4,0.05) {TSN-capable\\Edge Platform};
      \node[draw, fill=black!15, shape=cloud, aspect=1.7, align=center] (TSN2) at (-2,1) {Robot\\\scriptsize(TSN System)};
      \node[draw, fill=black!15, shape=cloud, aspect=1.7, align=center] (TSN3) at (-2,-0.9) {Robot\\\scriptsize(TSN System)};

      \node[draw, align=center, fill=black!30] (CNC) at (2.5,3.25) {TSN Centralized\\Network Controller};

      \draw (CNC) edge[<->,black!50,very thick] (5G);
      \draw (CNC) edge[<->,black!50,very thick,out=0,in=110] (TSN1);
      \draw (CNC) edge[<->,black!50,very thick,out=180,in=70] (TSN2);
      \draw (CNC) edge[<->,black!50,very thick,out=180,in=70] (TSN3);
  \end{tikzpicture}
  }
  \caption{Integration of 5G and TSN, as standardized by~\cite{3gpp.23.501}.}\label{fig:5g_tsn_bridge}
\end{figure}
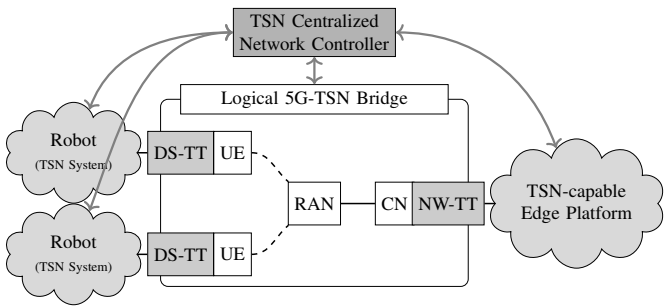

\subsection{The 5G System as a Logical TSN Bridge} \label{sec:5g_tsn}
To provide the same TSN mechanisms in networks that require both wired and wireless elements, e.g., free-moving robots in an industrial setting, 3GPP recently standardized an architecture to expose the 5G system as a logical TSN bridge~\cite{3gpp.23.501}. 
An overview of this integration is provided in Fig.~\ref{fig:5g_tsn_bridge}:
The TSN systems on the left show two moving robots, containing a single TSN end-station or a more complex internal network.
Each robot has a user equipment~(UE) to connect wirelessly to the 5G radio access network~(RAN).
When data is sent to the edge computing platform on the right, the RAN forwards the frames via the 5G core network (CN).

Instead of exposing the entire internal state (e.g., 5G resource allocation or session management), the 5G system provides device-side~(DS-TT) and network-side TSN translators~(NW-TT).
From the perspective of the TSN controller, the 5G system thus appears as a logical TSN bridge that supports the TSN mechanisms, as in Section~\ref{sec:tsn}.
Still, Fig.~\ref{fig:delays} shows an evident non-functional difference in the port-to-port delays of wired bridges and logical 5G-TSN bridges:
For wired bridges, the port-to-port delay corresponds to the processing delay and exhibits only small delay variations below $\qty{1}{\us}$.
For logical 5G-TSN bridges, however, the uplink port-to-port delay covers the entirety from the DS-TT transmission start until the NW-TT reception time, yielding millisecond delays and millisecond delay variations (a difference by three orders of magnitude).

The scheduler must account for this difference when synthesizing TSN configurations.
We therefore assume that the 5G system reports port-to-port delay histograms to the CNC, as in Fig.~\ref{fig:delays}(b).
Different streams can be served with different histograms.
For this work, we assume that the histograms are stationary and stay valid for any TSN scheduling decision.
However, even in this idealized setting, we show that existing approaches cannot provide any QoS guarantees or do not scale.

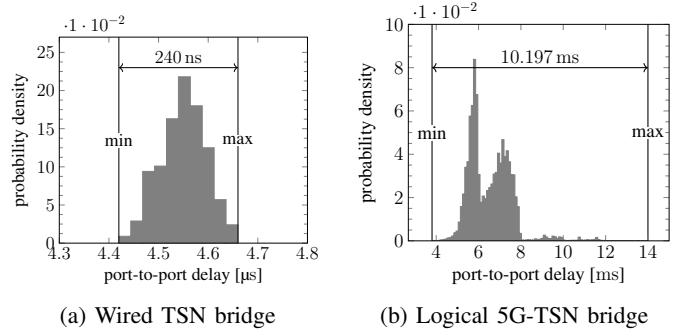
\begin{figure}
  \begin{subfigure}{0.48\columnwidth}
    \resizebox{\textwidth}{!}{%
      \begin{tikzpicture}
	  \begin{axis}[
	      ymin=0, ymax=0.27,
	      xmin=4.3,xmax=4.8,
	      minor y tick num = 3,
	      scaled y ticks=real:0.01,
	      area style,
	      xlabel={port-to-port delay [\si{\us}]},
	      ylabel={probability density},
	      font=\Large,
	      ]
	      \addplot+[ybar interval,mark=no,black,fill=black,semitransparent] plot coordinates { 
		  (4.42, 0.009) (4.444,0.029) (4.468,0.094) (4.492,0.101) (4.516,0.163) (4.54,0.218) (4.564,0.18) (4.588,0.125) (4.612,0.057) (4.636,0.024) (4.66,0)  
	      };
	      \draw (axis cs:4.42,0) edge[thick] node[inner sep=1pt,fill=white,pos=0.5] {min} (axis cs:4.42,0.27);
	      \draw (axis cs:4.66,0) edge[thick] node[inner sep=1pt,fill=white,pos=0.5] {max} (axis cs:4.66,0.27);
	      \draw (axis cs:4.42,0.23) edge[thick,<->] node[above] {$\qty{240}{\ns}$} (axis cs:4.66,0.23);
	  \end{axis}
      \end{tikzpicture}
    }
    \caption{Wired TSN bridge}
  \end{subfigure}
  \hfill
  \begin{subfigure}{0.48\columnwidth}
    \resizebox{\textwidth}{!}{%
      \begin{tikzpicture}
	  \begin{axis}[
	      ymin=0, ymax=0.1,
	      xmax=15,
	      minor y tick num = 3,
	      scaled y ticks=real:0.01,
	      area style,
	      xlabel={port-to-port delay [\si{\ms}]},
	      ylabel={probability density},
	      font=\Large,
	      ]
	      \addplot+[ybar interval,mark=no,black,fill=black,semitransparent] plot coordinates { 
		  (3.803,1e-05) (3.906,0.00011) (4.009,0.00011) (4.112,0.00029) (4.215,0.00041) (4.318,0.00055) (4.421,0.00088) (4.524,0.00136) (4.627,0.00161) (4.73,0.00224) (4.833,0.00372) (4.936,0.00452) (5.039,0.00866) (5.142,0.01582) (5.245,0.02163) (5.348,0.03399) (5.451,0.03653) (5.554,0.0576) (5.657,0.06254) (5.76,0.08382) (5.863,0.06751) (5.966,0.03065) (6.069,0.01778) (6.172,0.0207) (6.275,0.01947) (6.378,0.02322) (6.481,0.02457) (6.584,0.02623) (6.687,0.03085) (6.79,0.03164) (6.893,0.0405) (6.996,0.03587) (7.099,0.04674) (7.202,0.03896) (7.305,0.0416) (7.408,0.03652) (7.511,0.03066) (7.614,0.03047) (7.717,0.02029) (7.82,0.01577) (7.923,0.00345) (8.026,0.00041) (8.129,0.00026) (8.232,0.00044) (8.335,0.00048) (8.438,0.00033) (8.541,0.00041) (8.644,0.00093) (8.747,0.00057) (8.85,0.00075) (8.953,0.00116) (9.056,0.00262) (9.159,0.0018) (9.262,0.00104) (9.365,0.00103) (9.468,0.00219) (9.571,0.00187) (9.674,0.00109) (9.777,0.00183) (9.88,0.00148) (9.983,0.00039) (10.086,0.00017) (10.189,0.0002) (10.292,0.0002) (10.395,0.00028) (10.498,0.00024) (10.601,0.0003) (10.704,0.00171) (10.807,0.00045) (10.91,0.00027) (11.013,0.00037) (11.116,0.0005) (11.219,0.00033) (11.322,0.00046) (11.425,0.0007) (11.528,0.00158) (11.631,0.00057) (11.734,0.00023) (11.837,0.0002) (11.94,5e-05) (12.043,1e-05) (12.146,0.0) (12.249,1e-05) (12.352,1e-05) (12.455,0.0) (12.558,3e-05) (12.661,6e-05) (12.764,1e-05) (12.867,1e-05) (12.97,1e-05) (13.073,1e-05) (13.176,1e-05) (13.279,0.0) (13.382,0.0) (13.485,4e-05) (13.588,0.0) (13.691,2e-05) (13.794,0.0) (13.897,2e-05) (14.0,0.0) 
	      };
	      \draw (axis cs:3.803,0) edge[thick] node[inner sep=1pt,fill=white,pos=0.5] {min} (axis cs:3.803,0.1);
	      \draw (axis cs:14,0) edge[thick] node[inner sep=1pt,fill=white,pos=0.5] {max} (axis cs:14,0.1);
	      \draw (axis cs:3.803,0.08) edge[thick,<->] node[above] {$\qty{10.197}{\ms}$} (axis cs:14,0.08);
	  \end{axis}
      \end{tikzpicture}
    }
    \caption{Logical 5G-TSN bridge}
  \end{subfigure}
  \caption{Packet delay characteristics, measured by \cite{downlink_example_histogram}.} \label{fig:delays}
\end{figure}
\begin{figure*}[ht]
  \centering
  \resizebox{0.91\textwidth}{!}{%
  \begin{tikzpicture}
    \def\offset{0.5}
    \def\scale{0.35}
    \def\steps{2}

    \node at (-12, -1.5) {
      \begin{tikzpicture}
	\draw[black!60] (-0.25,-1.5) rectangle ++(2.5,3) node[above,xshift=-1.25cm] {\footnotesize Logical 5G-TSN Bridge};

      	\node[draw,fill=white,rounded corners,inner sep=5pt] (DS1) at (0,1) {$T^{DS}_1$};
      	\node[draw,fill=white,rounded corners,inner sep=5pt] (DS2) at (0,-1) {$T^{DS}_2$};
      	\node[draw,fill=white,rounded corners,inner sep=5pt] (NW) at (2,0) {$B^{NW}$};
      	\node[draw,rounded corners,inner sep=5pt] (B1) at (3.5,0) {$B_1$};
      	\node[draw,rounded corners,inner sep=5pt] (T3) at (3,-1.5) {$T_3$};
      	\node[draw,rounded corners,inner sep=5pt] (L1) at (5,1) {$L_1$};
      	\node[draw,rounded corners,inner sep=5pt] (L2) at (5,-1) {$L_2$};

	\draw[thick,dashed] (DS1) -- (NW);
	\draw[thick,dashed] (DS2) -- (NW);
	\draw[thick] (NW) -- (B1);
	\draw[thick] (T3) -- (B1);
	\draw[thick] (B1) -- (L1);
	\draw[thick] (B1) -- (L2);

	\draw[thick,->,F1,rounded corners=10pt] (DS1.-10) -- (1.6,0.4) -- node[above] {$F_1$} (3.65, 0.4) -- (L1);
	\draw[thick,->,F2,rounded corners=10pt] (DS2.10) -- (1.6,-0.4) -- node[below] {$F_2$} (3.65, -0.4) -- (L2);
	\draw[thick,->,F3,rounded corners=10pt] (T3) -- node[right] {$F_3$} (3.65, -0.4) -- (L2);

	\draw (0,-2.2) edge[thick,dashed] node[right,xshift=3mm] {5G links} (0.5,-2.2);
	\draw (2.5,-2.2) edge[thick] node[right,xshift=3mm] {Ethernet links} (3,-2.2);
      \end{tikzpicture}
    };

    \node[draw] (intended) at (0,0) {
      \resizebox{6cm}{!}{%
      \begin{tikzpicture}
	\draw[thick] (0,0) -- (6.5,0) node[pos=0,xshift=-1cm,yshift=0.375cm] {$[B_1,L_2]$};
	\draw[thick] (0,0.75) -- (6.5,0.75) node[pos=0,xshift=-1cm,yshift=0.375cm] {$[B_1,L_1]$};
	\draw[thick] (0,1.5) -- (6.5,1.5) node[pos=0,xshift=-1cm,yshift=0.375cm] {$[B^{NW},B_1]$};

	\draw[rounded corners, fill=f1] (\scale*0+\offset,1.5) rectangle ++(\scale*4,0.5) node[pos=.5] {$f_1$};
	\draw[rounded corners, fill=f2] (\scale*7+\offset,1.5) rectangle ++(\scale*4,0.5) node[pos=.5] {$f_2$};

	\draw[rounded corners, fill=f1] (\scale*4.5+\offset,0.75) rectangle ++(\scale*4,0.5) node[pos=.5] {$f_1$};
	\draw[rounded corners, fill=f3] (\scale*4.5+\offset,0) rectangle ++(\scale*4,0.5) node[pos=.5] {$f_3$};
	\draw[rounded corners, fill=f2] (\scale*11.5+\offset,0) rectangle ++(\scale*4,0.5) node[pos=.5] {$f_2$};
      \end{tikzpicture}
      }
    };

    \node[draw] (unintended) at (0,-3) {
      \resizebox{6cm}{!}{%
      \begin{tikzpicture}
	\draw[thick] (0,0) -- (6.5,0) node[pos=0,xshift=-1cm,yshift=0.375cm] {$[B_1,L_2]$};
	\draw[thick] (0,0.75) -- (6.5,0.75) node[pos=0,xshift=-1cm,yshift=0.375cm] {$[B_1,L_1]$};
	\draw[thick] (0,1.5) -- (6.5,1.5) node[pos=0,xshift=-1cm,yshift=0.375cm] {$[B^{NW},B_1]$};

	\draw[rounded corners,fill=f2] (\scale*0+\offset,1.5) rectangle ++(\scale*4,0.5) node[pos=.5] {$f_2$};
	\draw[rounded corners,fill=f1] (\scale*7+\offset,1.5) rectangle ++(\scale*4,0.5) node[pos=.5] {$f_1$};

	\draw[rounded corners] (\scale*4.5+\offset,0.75) rectangle ++(\scale*4,0.5) node[pos=.5] {empty};
	\draw[rounded corners, fill=f2] (\scale*4.5+\offset,0) rectangle ++(\scale*4,0.5) node[pos=.5] {$f_2$};
	\draw[rounded corners, fill=f3] (\scale*11.5+\offset,0) rectangle ++(\scale*4,0.5) node[pos=.5] {$f_3$};
      \end{tikzpicture}
      }
    };

    \node[above=1mm of intended,align=center] {Intended Transmission Ordering};
    \node[above=1mm of unintended] {Unfixable/Cascading Frame Reordering};
    \draw[-{Triangle[width=18pt,length=8pt]},black!20,line width=10pt] (-5.81,-1.5) edge[out=90,in=174] node[pos=0.75,above=1mm,black] {$f_1 \prec f_2$} (intended);
    \draw[-{Triangle[width=18pt,length=8pt]},black!20,line width=10pt] (-5.81,-1.5) edge[out=-90,in=-174] node[pos=0.75,below=1mm,black] {$f_2 \prec f_1$} (unintended);
    \draw[fill=white,white] (-7.31,-2.75) rectangle ++(3,2.5);

    \node (hist) at (-6,-1.65) {
      \resizebox{3.75cm}{!}{%
      \begin{tikzpicture}
          \begin{axis}[
              ymin=0, ymax=0.1,
              xmax=17,
              minor y tick num = 3,
              scaled y ticks=real:0.01,
              area style,
	      xlabel={arrival time at $B^{NW}$ [\si{\ms}]},
              ylabel={probability density},
              font=\Large,
              ]

	      \node (f1) at (axis cs: 3.5, 0.05) {$F_1$};
	      \draw (axis cs: 5.5,0.02) edge[out=180,in=0] (f1);

	      \node (f2) at (axis cs: 13, 0.05) {$F_2$};
	      \draw (axis cs: 10.5,0.02) edge[out=0,in=180] (f2);

              \addplot+[ybar interval,mark=no,f2,fill=f2] plot coordinates { 
        	  (6.803,1e-05) (6.906,0.00011) (7.009,0.00011) (7.112,0.00029) (7.215,0.00041) (7.318,0.00055) (7.421,0.00088) (7.524,0.00136) (7.627,0.00161) (7.73,0.00224) (7.833,0.00372) (7.936,0.00452) (8.039,0.00866) (8.142,0.01582) (8.245,0.02163) (8.348,0.03399) (8.451,0.03653) (8.554,0.0576) (8.657,0.06254) (8.76,0.08382) (8.863,0.06751) (8.966,0.03065) (9.069,0.01778) (9.172,0.0207) (9.275,0.01947) (9.378,0.02322) (9.481,0.02457) (9.584,0.02623) (9.687,0.03085) (9.79,0.03164) (9.893,0.0405) (9.996,0.03587) (10.099,0.04674) (10.202,0.03896) (10.305,0.0416) (10.408,0.03652) (10.511,0.03066) (10.614,0.03047) (10.717,0.02029) (10.82,0.01577) (10.923,0.00345) (11.026,0.00041) (11.129,0.00026) (11.232,0.00044) (11.335,0.00048) (11.438,0.00033) (11.541,0.00041) (11.644,0.00093) (11.747,0.00057) (11.85,0.00075) (11.953,0.00116) (12.056,0.00262) (12.159,0.0018) (12.262,0.00104) (12.365,0.00103) (12.468,0.00219) (12.571,0.00187) (12.674,0.00109) (12.777,0.00183) (12.88,0.00148) (12.983,0.00039) (13.086,0.00017) (13.189,0.0002) (13.292,0.0002) (13.395,0.00028) (13.498,0.00024) (13.601,0.0003) (13.704,0.00171) (13.807,0.00045) (13.91,0.00027) (14.013,0.00037) (14.116,0.0005) (14.219,0.00033) (14.322,0.00046) (14.425,0.0007) (14.528,0.00158) (14.631,0.00057) (14.734,0.00023) (14.837,0.0002) (14.94,5e-05) (15.043,1e-05) (15.146,0.0) (15.249,1e-05) (15.352,1e-05) (15.455,0.0) (15.558,3e-05) (15.661,6e-05) (15.764,1e-05) (15.867,1e-05) (15.97,1e-05) (16.073,1e-05) (16.176,1e-05) (16.279,0.0) (16.382,0.0) (16.485,4e-05) (16.588,0.0) (16.691,2e-05) (16.794,0.0) (16.897,2e-05) (17.0,0.0) 
              };
              \addplot+[ybar interval,mark=no,f1,fill=f1] plot coordinates { 
        	  (3.803,1e-05) (3.906,0.00011) (4.009,0.00011) (4.112,0.00029) (4.215,0.00041) (4.318,0.00055) (4.421,0.00088) (4.524,0.00136) (4.627,0.00161) (4.73,0.00224) (4.833,0.00372) (4.936,0.00452) (5.039,0.00866) (5.142,0.01582) (5.245,0.02163) (5.348,0.03399) (5.451,0.03653) (5.554,0.0576) (5.657,0.06254) (5.76,0.08382) (5.863,0.06751) (5.966,0.03065) (6.069,0.01778) (6.172,0.0207) (6.275,0.01947) (6.378,0.02322) (6.481,0.02457) (6.584,0.02623) (6.687,0.03085) (6.79,0.03164) (6.893,0.0405) (6.996,0.03587) (7.099,0.04674) (7.202,0.03896) (7.305,0.0416) (7.408,0.03652) (7.511,0.03066) (7.614,0.03047) (7.717,0.02029) (7.82,0.01577) (7.923,0.00345) (8.026,0.00041) (8.129,0.00026) (8.232,0.00044) (8.335,0.00048) (8.438,0.00033) (8.541,0.00041) (8.644,0.00093) (8.747,0.00057) (8.85,0.00075) (8.953,0.00116) (9.056,0.00262) (9.159,0.0018) (9.262,0.00104) (9.365,0.00103) (9.468,0.00219) (9.571,0.00187) (9.674,0.00109) (9.777,0.00183) (9.88,0.00148) (9.983,0.00039) (10.086,0.00017) (10.189,0.0002) (10.292,0.0002) (10.395,0.00028) (10.498,0.00024) (10.601,0.0003) (10.704,0.00171) (10.807,0.00045) (10.91,0.00027) (11.013,0.00037) (11.116,0.0005) (11.219,0.00033) (11.322,0.00046) (11.425,0.0007) (11.528,0.00158) (11.631,0.00057) (11.734,0.00023) (11.837,0.0002) (11.94,5e-05) (12.043,1e-05) (12.146,0.0) (12.249,1e-05) (12.352,1e-05) (12.455,0.0) (12.558,3e-05) (12.661,6e-05) (12.764,1e-05) (12.867,1e-05) (12.97,1e-05) (13.073,1e-05) (13.176,1e-05) (13.279,0.0) (13.382,0.0) (13.485,4e-05) (13.588,0.0) (13.691,2e-05) (13.794,0.0) (13.897,2e-05) (14.0,0.0) 
              };

          \end{axis}
      \end{tikzpicture}
      }
    };
  \end{tikzpicture}
  }
  \caption{Frame reordering in a TSN configuration for time-triggered streams can nullify any end-to-end QoS guarantee.}\label{fig:problem}
\end{figure*}
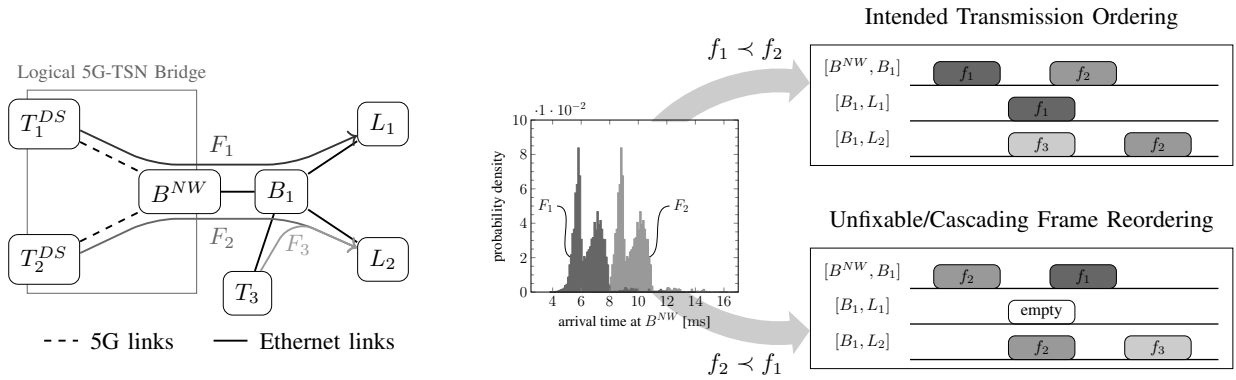

\subsection{System Model and Notation}
We model the network as a directed graph $G = (V, E)$.
A vertex $u \in V$ represents a network entity visible from the TSN controller (i.e., end-devices, wired TSN bridges, and TSN translators), whereas an edge $[u,v] \in E$ (i.e., $u,v \in V$) specifies an Ethernet or 5G link between two adjacent network entities.
Note that, in the case of Fig.~\ref{fig:5g_tsn_bridge}, we do not use a single vertex to model the logical 5G-TSN bridge;
instead, we use three vertices for the two DS-TTs and the NW-TT.
We argue that this way of modelling makes our results also applicable for other wireless technologies like IEEE 802.11.

We denote the set of (time-triggered) streams by $\TT$.
Each stream $F \in \TT$ defines its path $\path{F}$, period $\period{F}$, phase $\phase{F}$, and frame size $\size{F}$.
For this work, we assume unicast paths of the form $(\v{1}{F}, \ldots, \vl{F})$, where $\v{1}{F}$ is the talker and $\vl{F}$ is the listener. 
The hypercycle $H$ is defined by $\text{lcm}_{F \in \TT}(\period{F})$.
To find a feasible TSN configuration, the scheduler has to incorporate $H / \period{F}$ many frames of the stream $F$.
We denote frames of $F$ by $f \in F$, and the $i$th frame in a hypercycle is released at time
\begin{equation*}
  \release{f} = \phase{F} + i \times \period{F}.
\end{equation*}

\begin{definition}[End-to-End QoS] \label{def:reliability}
  Each stream $F \in \TT$ specifies its minimum end-to-end reliability $\rel{F}$ w.r.t. latency $\ete{F}$ and jitter $\jitter{F}$.
The TSN configuration must ensure, with a probability of at least $\rel{F}$, that each frame $f \in F$ arrives at its listener $\vl{F}$ within a predetermined arrival interval $\R{\vl{F}, f}$, as constrained by
\begin{align}
  \Rmax{\vl{F}, f} - \release{f} &\leq \ete{F} \qquad \text{and} \label{eq:latency} \\
  \sub{\Rmax}{\Rmin}{\vl{F}, f} &\leq \jitter{F}. \label{eq:jitter}
\end{align}
\end{definition}

This work aims to compute feasible TSN configurations $\Conf = (\GCL, \R{})$ that satisfy Definition~\ref{def:reliability}.
$\Conf$ defines the gating operations $\GCL([u,v])$ of the Time-Aware Shaper at each egress port $[u,v]$ and the allowed arrival interval $\R{\v{k}{F}, f}$ of PSFP for each frame $f \in F$ at each hop $\v{k}{F} \in \path{F}$.

\subsection{Modelling TSN Runtime Behavior} \label{sec:execution_sequences}
To analyze the QoS guarantees of a TSN configuration $\Conf$, we introduce a model that captures possible runtime behavior under $\Conf$.
To this end, we define execution sequences $\E = (\T{}, \D{})$ to capture the possible transmission offsets $\T{}$ and transmission delays $\D{}$ for each frame $f \in \TT$ and each hop $[u,v] \in \path{f}$.
In detail, $\E$ encodes the following semantic:
\begin{itemize}
  \item $\T{[u,v], f}$ denotes the time when the bridge $u$ starts the transmission of $f$ at the egress port $[u,v]$, and
  \item $\D{[u,v], f}$ denotes the delay until $f$ is enqueued at $v$.
\end{itemize}
For Ethernet links, $\D{[u,v], f}$ equals the sum of the serialization delay $\size{f} / \drate{[u,v]}$, the propagation delay $\dprop{[u,v]}$, and the processing delay $\dproc{v}$ of $v$.
For 5G links, $\D{[u,v], f}$ is the outcome of the random variable capturing the 5G port-to-port delay, as described in Section~\ref{sec:5g_tsn}.
In both cases, the arrival time of frame $f$ at bridge $v$ equals
\begin{equation*}
  \add{\T}{\D}{[u,v], f} = \T{[u,v], f} + \D{[u,v], f}.
\end{equation*}
Moreover, $\E$ is constrained to satisfy the following:

\textbf{FIFO Queueing:}
The transmission of a frame $f$ via $[u,v]$ can only start after $f$ arrived at $u$. 
If another frame $f'$ is enqueued in the same FIFO queue at $[u,v]$ before $f$, then $u$ can only start the transmission of $f$ once that of $f'$ is completed.

\textbf{Multiplexing:}
Only a single frame can be transmitted over an Ethernet link at once.
In contrast, frequency-division multiplexing allows multiple frames to be sent over 5G links in parallel.
As noted in Section~\ref{sec:5g_tsn}, we assume that the induced traffic load is captured by the 5G packet delay histograms.

\textbf{GCL Consistency:}
A frame can be transmitted via $[u,v]$ if the gate configured by $\GCL([u,v])$ is open.
At the same time, the first frame of the egress queue must immediately start its transmission once the gate opens.

\textbf{Stream Policing:}
A frame $f$ is discarded by PSFP at bridge $u$ if and only if $f$ arrives at $u$ outside the interval $\R{u, f}$.

\section{Problem Description} \label{sec:problem_description}
The integration of the 5G system in TSN as a logical 5G-TSN bridge enables IEEE 802.1Qbv scheduling in networks with wired and wireless network elements.
However, logical 5G-TSN bridges and wired TSN bridges differ significantly in their packet delay characteristics (see Fig.~\ref{fig:delays}).
Next, we show how this difference can effectively nullify the end-to-end QoS of streams and how it can induce a central bottleneck in scalability.
In particular, we argue that current IEEE 802.1Qbv scheduling approaches fall into the following two categories.

\textbf{Non-Robust Approaches}
are unable to capture 5G packet delay variations.
This can lead to unintended frame reordering and cascading QoS violations that spread throughout the entire network.
Most prominently, this category includes scheduling algorithms that consider scalar values for the 5G packet delays \ldash e.g., the average, median, or maximum delay of a 5G packet delay histogram (e.g., \cite{9212049, 9940254}).

This problem is illustrated in Fig.~\ref{fig:problem}, where three streams $F_1, F_2$, and $F_3$ (of the same priority) traverse the network as shown on the left.
The talkers $T_1^{DS}$ and $T_2^{DS}$ start the transmission of the frames $f_1 \in F_1$ and $f_2 \in F_2$ at times $\qty{0}{\ms}$ and $\qty{3}{\ms}$, respectively. 
With the 5G packet delay histograms of Fig.~\ref{fig:delays}(b), $f_1$ and $f_2$ can arrive at the NW-TT $B^{NW}$ within the intervals $[\qty{4}{\ms}, \qty{14}{\ms}]$ and $[\qty{7}{\ms}, \qty{17}{\ms}]$.

Even considering the latest possible arrival time ($\qty{14}{\ms}$ and $\qty{17}{\ms}$) will result in a schedule that breaks down irreparably after a short period of operation.
This is caused by the overlapping arrival intervals at $B^{NW}$, where a seemingly small probability of $f_2$ arriving before $f_1$ (denoted by $f_2 \prec f_1$ in Fig.~\ref{fig:problem}) has the following consequences:
\begin{enumerate}[label=(\roman*)]
  \item $f_1$ is forwarded later than intended and also misses its subsequent transmission slot at $[B_1, L_1]$,
  \item $f_2$ is forwarded earlier than intended and can steal the transmission slot of $f_3$ at $[B_1, L_2]$.
\end{enumerate}
Thus, even in such simplified settings, a single frame reordering has cascading effects that spread across multiple streams.
Our evaluations in Section~\ref{sec:eval} show that these effects cause a $\qty{99.99}{\percent}$ reliability in the inner 5G network to plummet below $\qty{10}{\percent}$ when looking at the end-to-end communication in TSN.

\textbf{Strict Transmission Isolation (STI).}
An apparent fix to the problems of non-robust approaches would be to allow no overlap in the arrival intervals of $f_1$ and $f_2$ at $B^{NW}$.
STI is often employed implicitly for IEEE 802.1Qbv scheduling in wired networks~\cite{nwps, Craciunas2016RTNS}.
However, compared to the small delay variations in wired networks, 5G packet delay variations are three orders of magnitude larger (cf. Fig.~\ref{fig:delays}).
Thus, STI would defer the initial transmission of $f_2$ by an additional $\qty{7}{\ms}$ to avoid an overlapping arrival interval with $f_1$ at $B^{NW}$.
This already prohibits $f_2$ from being accepted if its end-to-end latency must be below $\qty{20}{\ms}$.
Employing STI would therefore introduce a major scalability bottleneck in the number of schedulable wireless streams.

In this work, we aim to devise an IEEE 802.1Qbv scheduling technique that provides formal end-to-end QoS guarantees for each TSN stream (according to Definition~\ref{def:reliability}) without suffering the same scalability bottleneck as STI.
To quantify the scalability issue, we consider the number of schedulable TSN streams as our optimization objective.

\section{Robust and Feasible TSN Configurations} \label{sec:robustness_and_feasibility}
The previous section showed that highly stochastic 5G packet delays can nullify any QoS guarantee of non-robust approaches.
The underlying problem is rooted in the disconnect between the intended transmission ordering of the GCL and the actual transmission behavior at runtime. 
To overcome this obstacle, this section proposes a notion of robust scheduling that account for bounded packet delay uncertainty. 

\subsection{5G Packet Delay Budgets} \label{sec:robustness:pdb}
Each TSN stream $F$ specifies QoS requirements in terms of latency $\ete{F}$, jitter $\jitter{F}$, and reliability $\rel{F}$.
Due to significant 5G packet delays and packet delay variations, each of these metrics is strongly affected for wireless streams.
We therefore start by defining the allowed 5G packet delay budget~(PDB) that the TSN configuration has to account for.
For simplicity, we assume that each wireless stream traverses at most one logical 5G-TSN bridge and that there are no transmission faults in wired network partitions.\footnote{
  An extension to multiple wireless hops is straightforward by using (\ref{eq:pdb2}) multiplicatively for each 5G link.
  To address wired link/bridge failures, however, complementary techniques like frame replication are more suitable.
}

\subsubsection{Design Rationale}
A 5G PDB $\dpdb{} = [\dpdbmin{}, \dpdbmax{}]$ lower- and upper-bounds the allowed packet delay for traversing the logical 5G-TSN bridge. 
When regarding $\ete{F}$ as the end-to-end delay budget of $F$, the key question is: how much of that budget should be allocated to its 5G PDB $\dpdb{F}$?
We capture the trade-off between choosing $\dpdb{F}$ too narrow or too large by formulating the following optimization problem:
\begin{subequations} \label{eq:pdb}
\begin{align}
  \min_{\dpdb{}} \quad &\dpdbmax{F}, \label{eq:pdb1}\\
  \text{s.t.} \quad &\Prob\Big(\D{F} \in \dpdbfull{F}\Big) \geq \rel{F}. \label{eq:pdb2}
\end{align}
\end{subequations}
Here, $\D{F}$ denotes the random variable of 5G delays experienced by frames of stream $F$, and $\Prob(\cdot)$ denotes the corresponding probability function.
In the following, we describe the rationale behind the choice of (\ref{eq:pdb1}) and (\ref{eq:pdb2}).

\begin{figure}[t]
  \begin{subfigure}{0.48\columnwidth}
    \centering
    \resizebox{\textwidth}{!}{%
    \begin{tikzpicture}
	\begin{axis}[
	    ymin=0, ymax=0.1,
	    xmax=15,
	    minor y tick num = 3,
	    scaled y ticks=real:0.01,
	    area style,
	    xlabel={uplink packet delay [\si{\ms}]},
	    ylabel={probability density},
	    font=\Large,
	    ]
	    \path[name path=lower9999] (axis cs:3.803,0) -- (axis cs:3.803,1);
	    \path[name path=upper9999] (axis cs:13.176,0) -- (axis cs:13.176,1);
	    \addplot [black!10] fill between[of = lower9999 and upper9999];
	    \node[rotate=90,anchor=south] at (axis cs:13.176,0.08) {$\qty{99.99}{\percent}$};

	    \path[name path=lower99] (axis cs:3.803,0) -- (axis cs:3.803,1);
	    \path[name path=upper99] (axis cs:9.983,0) -- (axis cs:9.983,1);
	    \addplot [black!20] fill between[of = lower99 and upper99];
	    \node[rotate=90,anchor=south] at (axis cs:9.981,0.08) {$\qty{99}{\percent}$};

	    \path[name path=lower90] (axis cs:3.803,0) -- (axis cs:3.803,1);
	    \path[name path=upper90] (axis cs:7.717,0) -- (axis cs:7.717,1);
	    \addplot [black!30] fill between[of = lower90 and upper90];
	    \node[rotate=90,anchor=south] at (axis cs:7.717,0.08) {$\qty{90}{\percent}$};

	    \addplot+[ybar interval,mark=no,black,fill=black,semitransparent] plot coordinates { 
		(3.803,1e-05) (3.906,0.00011) (4.009,0.00011) (4.112,0.00029) (4.215,0.00041) (4.318,0.00055) (4.421,0.00088) (4.524,0.00136) (4.627,0.00161) (4.73,0.00224) (4.833,0.00372) (4.936,0.00452) (5.039,0.00866) (5.142,0.01582) (5.245,0.02163) (5.348,0.03399) (5.451,0.03653) (5.554,0.0576) (5.657,0.06254) (5.76,0.08382) (5.863,0.06751) (5.966,0.03065) (6.069,0.01778) (6.172,0.0207) (6.275,0.01947) (6.378,0.02322) (6.481,0.02457) (6.584,0.02623) (6.687,0.03085) (6.79,0.03164) (6.893,0.0405) (6.996,0.03587) (7.099,0.04674) (7.202,0.03896) (7.305,0.0416) (7.408,0.03652) (7.511,0.03066) (7.614,0.03047) (7.717,0.02029) (7.82,0.01577) (7.923,0.00345) (8.026,0.00041) (8.129,0.00026) (8.232,0.00044) (8.335,0.00048) (8.438,0.00033) (8.541,0.00041) (8.644,0.00093) (8.747,0.00057) (8.85,0.00075) (8.953,0.00116) (9.056,0.00262) (9.159,0.0018) (9.262,0.00104) (9.365,0.00103) (9.468,0.00219) (9.571,0.00187) (9.674,0.00109) (9.777,0.00183) (9.88,0.00148) (9.983,0.00039) (10.086,0.00017) (10.189,0.0002) (10.292,0.0002) (10.395,0.00028) (10.498,0.00024) (10.601,0.0003) (10.704,0.00171) (10.807,0.00045) (10.91,0.00027) (11.013,0.00037) (11.116,0.0005) (11.219,0.00033) (11.322,0.00046) (11.425,0.0007) (11.528,0.00158) (11.631,0.00057) (11.734,0.00023) (11.837,0.0002) (11.94,5e-05) (12.043,1e-05) (12.146,0.0) (12.249,1e-05) (12.352,1e-05) (12.455,0.0) (12.558,3e-05) (12.661,6e-05) (12.764,1e-05) (12.867,1e-05) (12.97,1e-05) (13.073,1e-05) (13.176,1e-05) (13.279,0.0) (13.382,0.0) (13.485,4e-05) (13.588,0.0) (13.691,2e-05) (13.794,0.0) (13.897,2e-05) (14.0,0.0) 
	    };

	    \draw (axis cs: 3.803,0) edge[thick] node[fill=white, inner sep=1pt] {min} (axis cs: 3.803,0.1);
	    \draw (axis cs: 14,0) edge[thick] node[fill=white, inner sep=1pt] {max} (axis cs: 14,0.1);
	\end{axis}
    \end{tikzpicture}
    }
  \end{subfigure}
  \begin{subfigure}{0.48\columnwidth}
    \centering
    \begin{tikzpicture}
      \def\offset{0.1}
      \def\scale{0.35}
      \def\steps{2}

      \draw (0,0) edge[thick,->] node[pos=0,below,font=\scriptsize,anchor=north west,xshift=-2pt] {DS-TT} (3.5,0);
      \draw (0,-1) edge[thick,->] node[pos=0,below,font=\scriptsize,anchor=north west,xshift=-2pt] {NW-TT (PSFP)} (3.5,-1);
      \draw (0,-2) edge[thick,->] node[pos=0,below,font=\scriptsize,anchor=north west,xshift=-2pt] {NW-TT (egress queue)} (3.5,-2);

      \draw[rounded corners=1mm] (\scale*0+\offset,0) rectangle ++(\scale*1,0.4) node[pos=0.5,font=\scriptsize] {$f$};
      \draw (\scale*1+\offset,0.5) edge[thick,-{Latex[length=1.5mm]}] node[above=2mm,font=\scriptsize] {$tx$} (\scale*1+\offset,0);

      \draw [pattern=north east lines] (0,-1) rectangle ++(\scale*4+\offset,0.5);
      \draw (\scale*4+\offset,0.5) edge[dashed] node[above=9mm,xshift=7pt,font=\scriptsize,rotate=30] {$tx + \dpdbmin{}$} (\scale*4+\offset,-1);
      \draw (\scale*4+\offset,-0.5) edge[thick] (\scale*4+\offset,-1);
      \draw (\scale*8+\offset,0.5) edge[dashed] node[above=9mm,xshift=7pt,font=\scriptsize,rotate=30] {$tx + \dpdbmax{}$} (\scale*8+\offset,-1);
      \draw (\scale*8+\offset,-0.5) edge[thick] (\scale*8+\offset,-1);
      \draw [pattern=north east lines] (\scale*8+\offset,-1) rectangle ++(\scale*0.35+\offset,0.5) node[right,pos=0.5,xshift=2pt,font=\scriptsize] {$\cdots$};

      \draw[rounded corners=1mm] (\scale*5+\offset,-2) rectangle ++(\scale*3.5,0.4) node[pos=0.5,font=\scriptsize] {$f$};
      \draw[black!40] (\scale*1+\offset,0) -- (\scale*5+\offset,-0.5) -- (\scale*5+\offset,-2);
      \draw (\scale*5+\offset,-0.5) edge[thick,-{Latex[length=1.5mm]}] node[above=2mm,font=\scriptsize] {$rx$} (\scale*5+\offset,-1);
      \draw (\scale*8.5+\offset,-1.5) edge[thick,-{Latex[length=1.5mm]}] node[above=2mm,font=\scriptsize] {$tx$} (\scale*8.5+\offset,-2);
      
      \node at (0,-2.5) {};
    \end{tikzpicture}
  \end{subfigure}
  \caption{5G Packet Delay Budgets $[\mathbf{d}^{min}, \mathbf{d}^{max}]$ from a 5G perspective (left) and from a TSN perspective (right).} \label{fig:pdb}
\end{figure}
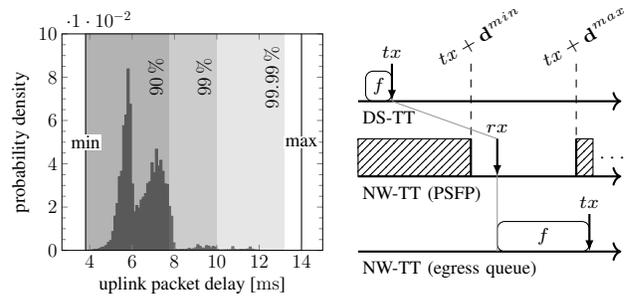

First, allocating a too narrow 5G PDB, either in its upper bound $\dpdbmax{}$ or its size $\dpdbmax{} - \dpdbmin{}$, results in a degradation of the stream's end-to-end reliability and increases the difficulty for the 5G system to allocate sufficient resources.
Moreover, for the uplink stream $F$ in Fig.~\ref{fig:pdb}, the expected arrival of a frame $f \in F$ at the NW-TT is given by
\begin{equation*}
  [tx + \dpdbmin{F}, tx + \dpdbmax{F}],
\end{equation*}
where $tx$ is the transmission offset of $f$ at the DS-TT.
As any arrival of $f$ outside this interval will likely result in uncontrolled frame reordering, PSFP at the NW-TT must discard $f$ in such cases.
Therefore, constraint~(\ref{eq:pdb2}) ensures that the 5G PDB of $F$ is large enough so that the probability of frames $f \in F$ passing PSFP is greater than $\rel{F}$.

Second, allocating a too large 5G PDB increases the difficulty for the IEEE 802.1Qbv scheduler to deliver frames within the allowed end-to-end latency.
For instance, Fig.~\ref{fig:pdb} shows how the PDB for a stream $F$ with $\rel{F} = \qty{90}{\percent}$ can cut off the long tail of the packet delay histogram, reducing the end-to-end latency by over $\qty{6}{\ms}$.
With the objective function~(\ref{eq:pdb1}), we therefore opt to minimize $\dpdbmax{F}$.

\subsubsection{5G PDB Contracts}
With the port-to-port delay histograms reported by the 5G system, the CNC can compute the 5G PDB of $F$ by accumulating the bin counts until $\rel{F}$ is exceeded. 
In detail, we can solve (\ref{eq:pdb}) optimally with
\begin{align*}
  \dpdbmin{F} &= \text{hist}[0].\text{low}, \\
  \dpdbmax{F} &= \min\Big\{\text{hist}[i].\text{up} \mid \sum\nolimits_{j=0}^{i} \text{hist}[j].\text{count} \geq \rel{F}\Big\},
\end{align*}
where $\text{hist}[i]$ denotes the $i$th bin of the histogram, with its respective lower/upper bound and (normalized) bin count.

The 5G PDB of a stream $F$ yields the following contract between the 5G system and the TSN system:
``The delay of a frame $f \in F$ traversing the 5G system is bounded by $\dpdb{F}$, with a probability greater than $\rel{F}$.''
The IEEE 802.1Qbv scheduler solely requires this contract to compute the GCLs.
That is, the CNC does not require additional state information of the 5G system (e.g., the current channel quality indicator) and it does not have to reschedule the entire TSN domain with every change in the 5G network state (e.g., which can occur every few milliseconds).
Instead, the 5G system can internally adapt its resource allocation to uphold the PDB contract.

\subsection{Robust and Feasible TSN Configurations} \label{sec:robustness}
With 5G PDBs capturing the guarantees of the 5G system, TSN has to extend these guarantees end-to-end from talker to listener.
This section introduces the notion of robustness that guarantees the frame's punctual delivery to its TSN listener as long as the 5G system complies with the 5G PDB contracts.

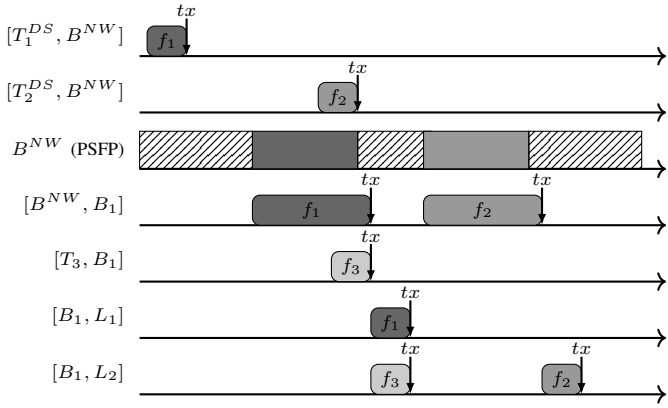
\begin{figure}
  \centering
  \tikzset{
  diagonal fill/.style 2 args={fill=#2, path picture={
  \fill[#1, sharp corners] (path picture bounding box.south west) -|
			   (path picture bounding box.north east) -- cycle;}},
  reversed diagonal fill/.style 2 args={fill=#2, path picture={
  \fill[#1, sharp corners] (path picture bounding box.north west) |- 
			   (path picture bounding box.south east) -- cycle;}}
  }
  \begin{tikzpicture}
      \def\offset{0.1}
      \def\scale{0.35}
      \def\steps{2}
      \def\y{-0.75}

      \draw (0,\y*0) edge[thick,->] node[pos=0,left,font=\scriptsize,anchor=east,xshift=-2pt,yshift=8pt] {$[T^{DS}_1, B^{NW}]$} (7,\y*0);
      \draw (0,\y*1) edge[thick,->] node[pos=0,left,font=\scriptsize,anchor=east,xshift=-2pt,yshift=8pt] {$[T^{DS}_2, B^{NW}]$} (7,\y*1);
      \draw (0,\y*2) edge[thick,->] node[pos=0,below,font=\scriptsize,anchor=east,xshift=-2pt,yshift=8pt] {$B^{NW}$ (PSFP)} (7,\y*2);
      \draw (0,\y*3) edge[thick,->] node[pos=0,below,font=\scriptsize,anchor=east,xshift=-2pt,yshift=8pt] {$[B^{NW}, B_1]$} (7,\y*3);
      \draw (0,\y*4) edge[thick,->] node[pos=0,below,font=\scriptsize,anchor=east,xshift=-2pt,yshift=8pt] {$[T_3, B_1]$} (7,\y*4);
      \draw (0,\y*5) edge[thick,->] node[pos=0,below,font=\scriptsize,anchor=east,xshift=-2pt,yshift=8pt] {$[B_1, L_1]$} (7,\y*5);
      \draw (0,\y*6) edge[thick,->] node[pos=0,below,font=\scriptsize,anchor=east,xshift=-2pt,yshift=8pt] {$[B_1, L_2]$} (7,\y*6);

      \draw[rounded corners=1mm, fill=f1] (\scale*0+\offset,\y*0) rectangle ++(\scale*1.5,0.4) node[pos=0.5,font=\scriptsize] {$f_1$};
      \draw (\scale*1.5+\offset,\y*0+0.5) edge[thick,-{Latex[length=1.5mm]}] node[above=1.5mm,font=\scriptsize] {$tx$} (\scale*1.5+\offset,\y*0);

      \draw[rounded corners=1mm, fill=f2] (\scale*6.5+\offset,\y*1) rectangle ++(\scale*1.5,0.4) node[pos=0.5,font=\scriptsize] {$f_2$};
      \draw (\scale*8+\offset,\y*1+0.5) edge[thick,-{Latex[length=1.5mm]}] node[above=1.5mm,font=\scriptsize] {$tx$} (\scale*8+\offset,\y*1);

      \draw[pattern=north east lines] (0,\y*2) rectangle ++(\scale*4+\offset,0.5);
      \draw[fill=f1] (\scale*4+\offset,\y*2) rectangle ++(\scale*4,0.5);  
      \draw[pattern=north east lines] (\scale*8+\offset,\y*2) rectangle ++(\scale*2.5+\offset,0.5);
      \draw[fill=f2] (\scale*10.5+\offset,\y*2) rectangle ++(\scale*4,0.5);  
      \draw[pattern=north east lines] (\scale*14.5+\offset,\y*2) rectangle ++(\scale*4+\offset,0.5);

      \draw[rounded corners=1mm, fill=f1] (\scale*4+\offset,\y*3) rectangle ++(\scale*4.5,0.4) node[pos=0.5,font=\scriptsize] {$f_1$};
      \draw (\scale*8.5+\offset,\y*3+0.5) edge[thick,-{Latex[length=1.5mm]}] node[above=1.5mm,font=\scriptsize] {$tx$} (\scale*8.5+\offset,\y*3);
      \draw[rounded corners=1mm, fill=f2] (\scale*10.5+\offset,\y*3) rectangle ++(\scale*4.5,0.4) node[pos=0.5,font=\scriptsize] {$f_2$};
      \draw (\scale*15+\offset,\y*3+0.5) edge[thick,-{Latex[length=1.5mm]}] node[above=1.5mm,font=\scriptsize] {$tx$} (\scale*15+\offset,\y*3);

      \draw[rounded corners=1mm, fill=f3] (\scale*7+\offset,\y*4) rectangle ++(\scale*1.5,0.4) node[pos=0.5,font=\scriptsize] {$f_3$};
      \draw (\scale*8.5+\offset,\y*4+0.5) edge[thick,-{Latex[length=1.5mm]}] node[above=1.5mm,font=\scriptsize] {$tx$} (\scale*8.5+\offset,\y*4);

      \draw[rounded corners=1mm, fill=f1] (\scale*8.5+\offset,\y*5) rectangle ++(\scale*1.5,0.4) node[pos=0.5,font=\scriptsize] {$f_1$};
      \draw (\scale*10+\offset,\y*5+0.5) edge[thick,-{Latex[length=1.5mm]}] node[above=1.5mm,font=\scriptsize] {$tx$} (\scale*10+\offset,\y*5);

      \draw[rounded corners=1mm, fill=f3] (\scale*8.5+\offset,\y*6) rectangle ++(\scale*1.5,0.4) node[pos=0.5,font=\scriptsize] {$f_3$};
      \draw (\scale*10+\offset,\y*6+0.5) edge[thick,-{Latex[length=1.5mm]}] node[above=1.5mm,font=\scriptsize] {$tx$} (\scale*10+\offset,\y*6);
      \draw[rounded corners=1mm, fill=f2] (\scale*15+\offset,\y*6) rectangle ++(\scale*1.5,0.4) node[pos=0.5,font=\scriptsize] {$f_2$};
      \draw (\scale*16.5+\offset,\y*6+0.5) edge[thick,-{Latex[length=1.5mm]}] node[above=1.5mm,font=\scriptsize] {$tx$} (\scale*16.5+\offset,\y*6);
  \end{tikzpicture}
  \caption{Strict temporal isolation of all transmissions yields a robust (yet inefficient) TSN configuration.
    The streams traverse the network as in Fig.~\ref{fig:problem}.
    The $tx$-annotated arrows show the frames' transmission start (i.e., the gate openings), and the leading boxes show their maximum queueing delay.
  }\label{fig:strict_isolation}
\end{figure}

Before formalizing robustness in Definition~\ref{def:robustness}, we describe the underlying idea with Fig.~\ref{fig:strict_isolation}.
To protect the QoS of $F_2$, the correct delivery of $f_2 \in F_2$ must be independent of potential transmission faults of $f_1$ and $f_3$.
Robustness must therefore ensure, for every execution sequence $\E$, that $f_2$ can only take its intended transmission slots, i.e., the second gate opening at $[B^{NW}, B_1]$ and at $[B_1,L_2]$.
This assurance is twofold: 
First, there must not be a queueing backlog that defers $f_2$.
Second, $f_2$ may never take an earlier transmission slot, even if the previous slot is empty (e.g., when $f_1$ violates its 5G PDB).

At the same time, robustness should be considered as a conditional guarantee that builds on the 5G PDB contracts of Section~\ref{sec:robustness:pdb}.
For instance, $f_2$ can only take its intended transmission slots if the 5G packet delay of $f_2$ stays within its PDB $\dpdb{[T_2^{DS}, B^{NW}], F_2}$;
otherwise, $B^{NW}$ may have to discard $f_2$ to protect other streams.
We extend this idea inductively in the number of transmission hops:
$f_2$ arrives at its $k$th hop within the expected arrival interval if all individual packet delays up to that point lie within their respective PDBs.\footnote{To unify notation, we extend the notion of 5G PDBs to Ethernet links, e.g., to $\dpdb{[B_1, L_2], F_2}$, which may capture small delay variations as in Fig.~\ref{fig:delays}(a).}
This leads to the following definition of robustness:

\begin{definition}[Robustness] \label{def:robustness}
  The TSN configuration $\Conf = (\GCL, \R{})$ robustly schedules a stream $F \in \TT$ if for every execution sequence $\E = (\T{}, \D{})$, every frame $f \in F$, and every hop $\v{k}{F} \in \path{F}$ the following holds true:
  If, up to bridge~$\v{k}{F}$, the packet delays lie within their PDBs, i.e.,
  \begin{equation*} \label{eq:robustness:precondition}
    \D{[\v{l}{F}, \v{l+1}{F}], f} \in \dpdb{[\v{l}{F}, \v{l+1}{F}], F}, \quad \forall 1 \leq l < k,
  \end{equation*}
  then $\v{k}{F}$ receives $f$ within its expected PSFP interval, i.e.,
  \begin{equation*} \label{eq:robustness:postcondition}
    \add{\T}{\D}{[\v{k-1}{F}, \v{k}{F}], f} \in \R{\v{k}{F}, f}.
  \end{equation*}
\end{definition}

Finally, we summarize the required properties of a TSN configuration $\Conf$ to achieve formal end-to-end QoS guarantees:
\begin{theorem}[Feasibility] \label{theorem:feasibility}
  A TSN configuration $\Conf$ feasibly schedules a stream $F \in \TT$ (according to Definition~\ref{def:reliability}) if
  \begin{enumerate}[label=(\roman*)]
    \item $\Conf$ robustly schedules $F$,
    \item $\Conf$ allocates sufficiently large PDBs, according to (\ref{eq:pdb}), and
    \item $\Conf$ satisfies the latency~(\ref{eq:latency}) and jitter~(\ref{eq:jitter}) bounds of $F$.
  \end{enumerate}
\end{theorem}
\begin{proof}
  The proof is given in Appendix~\ref{appendix:feasibility}.
\end{proof}

\noindent For example, combining \textit{(i)} and \textit{(ii)} ensures \ldash with a probability of at least $\rel{F_2}$ \rdash that $f_2$ arrives at its listener $L_2$ within the expected arrival interval $\R{L_2, f_2} = [rx_1, rx_2]$.
Moreover, \textit{(iii)} constrains $rx_2 \leq \ete{F_2}$ and $rx_2 - rx_1 \leq \jitter{F_2}$.

\section{Computing Feasible TSN Configurations} \label{sec:fips}
We introduce Full Interleaving Packet Scheduling~(FIPS) to compute feasible TSN configurations $\Conf$, without residing to strict temporal isolation.
In the following, we assume that the CNC already allocated the 5G PDBs, as presented in Section~\ref{sec:robustness:pdb}.
We therefore focus on computing robust TSN configurations that satisfy the latency and jitter constraints of all accepted TSN streams.
We start with an overview.

\subsection{Overview of FIPS}
Given a stream set $\TT$, we follow an incremental scheduling approach that adds one stream $F_j \in \TT$ to an existing TSN configuration (initially empty) at a time.
At the core of the IEEE 802.1Qbv scheduling problem, the scheduler has to decide the transmission order of frames $f_{i_1} \prec f_{i_2} \prec \ldots \prec f_{i_n}$ at each egress port~$[u,v]$.
To reduce the scalability bottleneck of strict temporal isolation, we allow batched frame transmissions over $[u,v]$ that result in a transmission ordering of the form
\begin{equation*}
  B_1 = \{f_{i_1}, f_{i_2}, \ldots, f_{i_{k_1}}\} \prec B_2 \prec \ldots \prec B_n.
\end{equation*}
The operation of FIPS can therefore be split in two parts:
Section~\ref{sec:fips:incremental} presents how to add a new stream $F_j$ to an existing transmission ordering, based on the QoS requirements of $F_j$ and the (feasible) TSN configuration $\Conf_{j-1}$ of the previous iteration.
Section~\ref{sec:fips:configuration} then shows how to derive a robust TSN configuration $\Conf_j$ from the new transmission ordering.
FIPS accepts the stream $F_j$ if and only if $\Conf_j$ satisfies the latency and jitter constraints of $F_j$ and all previously accepted streams.

\subsection{An Incremental Heuristic for Transmission Orderings} \label{sec:fips:incremental}
Let $\TT_{j-1} = \{F_1, \ldots, F_{j-1}\}$ be the set of streams that are already feasibly scheduled under the FIPS configuration $\Conf_{j-1}$.
To add a new stream $F$, we adjust the transmission orderings
\begin{equation*}
  B_1^k \prec B_2^k \prec \ldots \prec B_n^k
\end{equation*}
at each port $[\v{k}{F}, \v{k+1}{F}] \in \path{F}$ by inserting all frames $f \in F$ that are sent within one hypercycle.
We start by \textit{inserting} the frames in the transmission ordering as a new batch, and then check if the batch should be \textit{merged} with its immediate transmission predecessor or successor. 
Thus, the transmission ordering encodes the scheduling decisions, which is later used in Section~\ref{sec:fips:configuration} to derive a robust TSN configuration.

\textbf{Inserting Frames:}
For each frame $f \in F$, we compute a lower bound for the earliest possible transmission at $\v{k}{F}$ with
\begin{equation*}
  \phi([\v{k}{F}, \v{k+1}{F}], f) = \release{f} + \sum_{i = 1}^{k-1} \dpdbmax{[\v{i}{F}, \v{i+1}{F}], f}.
\end{equation*}
Similarly, let $\S{[\v{k}{F}, \v{k+1}{F}], B_i^k}$ denote the transmission start of the batch $B_i^k$ under the previous TSN configuration $\Conf_{j-1}$.
FIPS selects the maximum $i \in \{1, \ldots, n\}$ with
\begin{equation} \label{eq:insertion:1}
  \S{[\v{k}{F}, \v{k+1}{F}], B_i^k} \leq \phi([\v{k}{F}, \v{k+1}{F}], f)
\end{equation}
and insert $f$ in the transmission ordering as $B_i^k \prec f \prec B_{i+1}^k$.
Moreover, to avoid inconsistent transmission orderings, we fix the ordering $B_i^k \prec f$ for all consecutive transmission hops where $B_i^k$ and $f$ share the same FIFO queue.

\begin{table}
  \centering
  \caption{Policy to avoid FIFO inconsistencies.} \label{tab:inconsistency}
  \begin{tabular}{rccc}
    \toprule
     & $[T_1^{DS}, B^{NW}]$ & $[B^{NW}, B_1]$ & $[B_1, L_1]$ \\
    \midrule
    $\S{[u,v], f_1}$ & $\qty{0}{\ms}$ & $\qty{14}{\ms}$ & $\qty{14.01}{\ms}$ \\
    $\phi([u,v], f)$ & $\qty{1}{\ms}$ & $\qty{11}{\ms}$ & $\qty{11.01}{\ms}$ \\[1.5mm]
    Ordering of (\ref{eq:insertion:1}) & $f_1 \prec f$ & $\;\;f\; \prec \, f_1$ & $\;\;f\; \prec \, f_1$ \\
    FIFO Policy & $f_1 \prec f$ & $f_1 \prec f$ & $f_1 \prec f$ \\
    \bottomrule
  \end{tabular}
\end{table}

Table~\ref{tab:inconsistency} shows the calculations for inserting a new stream $F$ into a transmission ordering that already contains $F_1$ from our previous examples (e.g., see Fig.~\ref{fig:problem}).
$F$ has the same specification as $F_1$, but has a phase offset of $\qty{1}{\ms}$ and a smaller 5G PDB with $\dpdbmax{} = \qty{10}{\ms}$.
Solely applying~(\ref{eq:insertion:1}) would imply that the frame $f \in F$ has to overtake $f_1 \in F_1$, which is impossible if both frames share the same FIFO queues along $[T_1^{DS}, B^{NW}, B_1, L_1]$.
We therefore prioritize the already accepted $F_1$ by fixing the transmission ordering $f_1 \prec f$ for all consecutive transmission hops.

\textbf{Merging Batches.}
For wired streams $F$, there is no need for batching and the ordering of the previous step is returned. 
Otherwise, let $\v{k}{F}$ be the TSN translator that follows the 5G link.
With a 5G PDB of $\dpdbfull{}$, FIPS opts to de-jitter frames $f \in F$ at $\v{k}{F}$ for a minimum duration of 
\begin{equation*}
  \sub{\dpdbmax}{\dpdbmin}{[\v{k-1}{F}, \v{k}{F}], f}.
\end{equation*}
Recall that this duration can be in the range of milliseconds, where a strict transmission isolation leads to a major scalability bottleneck.
Hence, when $B_i^k \prec f \prec B_{i+1}^k$ is the transmission ordering of the previous step, FIPS considers three options:
merge $f$ with $B_i^k$, merge $f$ with $B_{i+1}^k$, or return the ordering as it is.
FIPS derives a robust TSN configuration, as described next, for all three options and verifies the QoS guarantees for all streams in $\TT_{j-1} \cup \{F\}$.

\subsection{Robust TSN Configurations from Transmission Orderings} \label{sec:fips:configuration}
Next, we present how to efficiently derive a FIPS configuration $\Conf$ from a given transmission ordering and prove that every FIPS configuration is robust (Theorem~\ref{theorem:fips_robust}).
By checking the latency (\ref{eq:latency}) and jitter (\ref{eq:jitter}) constraints for each frame, the procedure returns successfully if and only if $\Conf$ is feasible.

We start by defining the accumulated PDBs of a frame batch $B_i$.
For Ethernet links, the transmission delay is the summed delays of all frames $f \in B_i$, i.e.,
\begin{equation*}
  \dpdbtransmin{[u,v], B_i} = \dpdbtransmax{[u,v], B_i} = \sum_{f \in B_i} \frac{\size{f}}{\drate{[u,v]}},
\end{equation*}
whereas the total delay $\dpdb{[u,v], B_i}$ includes the propagation delay $\dprop{[u,v]}$ and the processing delay $\dproc{v}$ once.
For 5G links, frequency-division multiplexing allows the simultaneous transmission of all frames in $B_i$, under the delay conditions captured by the respective histograms;
hence, $\dpdb{[u,v], B_i}$ is set to the minimum and maximum bounds of the individual 5G PDBs $\dpdb{[u,v], f}$, for $f \in B_i$.

FIPS continues by determining the earliest possible transmission start $\S{[u,v], B_i}$ for each link $[u,v]$ and each batch $B_i$ by enforcing the following constraints C1--C3: 

\begin{figure}
  \centering
  \tikzset{
  diagonal fill/.style 2 args={fill=#2, path picture={
  \fill[#1, sharp corners] (path picture bounding box.south west) -|
			   (path picture bounding box.north east) -- cycle;}},
  reversed diagonal fill/.style 2 args={fill=#2, path picture={
  \fill[#1, sharp corners] (path picture bounding box.north west) |- 
			   (path picture bounding box.south east) -- cycle;}}
  }
  \begin{tikzpicture}
      \def\offset{0.1}
      \def\scale{0.35}
      \def\steps{2}
      \def\y{-0.75}

      \draw (0,\y*0) edge[thick,->] node[pos=0,left,font=\scriptsize,anchor=east,xshift=-2pt,yshift=8pt] {$[T^{DS}_1, B^{NW}]$} (7,\y*0);
      \draw (0,\y*1) edge[thick,->] node[pos=0,left,font=\scriptsize,anchor=east,xshift=-2pt,yshift=8pt] {$[T^{DS}_2, B^{NW}]$} (7,\y*1);
      \draw (0,\y*2) edge[thick,->] node[pos=0,below,font=\scriptsize,anchor=east,xshift=-2pt,yshift=8pt] {$B^{NW}$ (PSFP)} (7,\y*2);
      \draw (0,\y*3) edge[thick,->] node[pos=0,below,font=\scriptsize,anchor=east,xshift=-2pt,yshift=8pt] {$[B^{NW}, B_1]$} (7,\y*3);
      \draw (0,\y*4) edge[thick,->] node[pos=0,below,font=\scriptsize,anchor=east,xshift=-2pt,yshift=8pt] {$[T_3, B_1]$} (7,\y*4);
      \draw (0,\y*5) edge[thick,->] node[pos=0,below,font=\scriptsize,anchor=east,xshift=-2pt,yshift=8pt] {$[B_1, L_1]$} (7,\y*5);
      \draw (0,\y*6) edge[thick,->] node[pos=0,below,font=\scriptsize,anchor=east,xshift=-2pt,yshift=8pt] {$[B_1, L_2]$} (7,\y*6);

      \draw[rounded corners=1mm, fill=f1] (\scale*0+\offset,\y*0) rectangle ++(\scale*1.5,0.4) node[pos=0.5,font=\scriptsize] {$f_1$};
      \draw (\scale*1.5+\offset,\y*0+0.5) edge[thick,-{Latex[length=1.5mm]}] node[above=1.5mm,font=\scriptsize] {$tx$} (\scale*1.5+\offset,\y*0);

      \draw[rounded corners=1mm, fill=f2] (\scale*0+\offset,\y*1) rectangle ++(\scale*1.5,0.4) node[pos=0.5,font=\scriptsize] {$f_2$};
      \draw (\scale*1.5+\offset,\y*1+0.5) edge[thick,-{Latex[length=1.5mm]}] node[above=1.5mm,font=\scriptsize] {$tx$} (\scale*1.5+\offset,\y*1);

      \draw[pattern=north east lines] (0,\y*2) rectangle ++(\scale*4+\offset,0.5);
      \draw[diagonal fill={f2}{f1}] (\scale*4+\offset,\y*2) rectangle ++(\scale*4,0.5);  
      \draw[pattern=north east lines] (\scale*8+\offset,\y*2) rectangle ++(\scale*10.5+\offset,0.5);

      \draw[rounded corners=1mm, diagonal fill={f2}{f1}] (\scale*4+\offset,\y*3) rectangle ++(\scale*4.5,0.4) node[pos=0.5,font=\scriptsize] {$\{f_1,f_2\}$};
      \draw (\scale*8.5+\offset,\y*3+0.5) edge[thick,-{Latex[length=1.5mm]}] node[above=1.5mm,font=\scriptsize] {$tx$} (\scale*8.5+\offset,\y*3);
      \draw[rounded corners=1mm, fill=f3] (\scale*8.5+\offset,\y*4) rectangle ++(\scale*1.5,0.4) node[pos=0.5,font=\scriptsize] {$f_3$};
      \draw (\scale*10+\offset,\y*4+0.5) edge[thick,-{Latex[length=1.5mm]}] node[above=1.5mm,font=\scriptsize] {$tx$} (\scale*10+\offset,\y*4);

      \draw[rounded corners=1mm, fill=f1] (\scale*8.5+\offset,\y*5) rectangle ++(\scale*1.5,0.4) node[pos=0.5,font=\scriptsize] {$f_1$};
      \draw (\scale*10+\offset,\y*5+0.5) edge[thick,-{Latex[length=1.5mm]}] node[above=1.5mm,font=\scriptsize] {$tx$} (\scale*10+\offset,\y*5);

      \draw[rounded corners=1mm, fill=f2] (\scale*8.5+\offset,\y*6) rectangle ++(\scale*1.5,0.4) node[pos=0.5,font=\scriptsize] {$f_2$};
      \draw (\scale*10+\offset,\y*6+0.5) edge[thick,-{Latex[length=1.5mm]}] node[above=1.5mm,font=\scriptsize] {$tx$} (\scale*10+\offset,\y*6);
      \draw[rounded corners=1mm, fill=f3] (\scale*10.1+\offset,\y*6) rectangle ++(\scale*1.5,0.4) node[pos=0.5,font=\scriptsize] {$f_3$};
      \draw (\scale*11.6+\offset,\y*6+0.5) edge[thick,-{Latex[length=1.5mm]}] node[above=1.5mm,font=\scriptsize] {$tx$} (\scale*11.6+\offset,\y*6);
  \end{tikzpicture}
  \caption{FIPS configuration when batching $\{f_1, f_2\}$ at $[B^{NW}, B_1]$ and choosing $f_2 \prec f_3$ at $[B_1, L_2]$.} \label{fig:fips}
\end{figure}
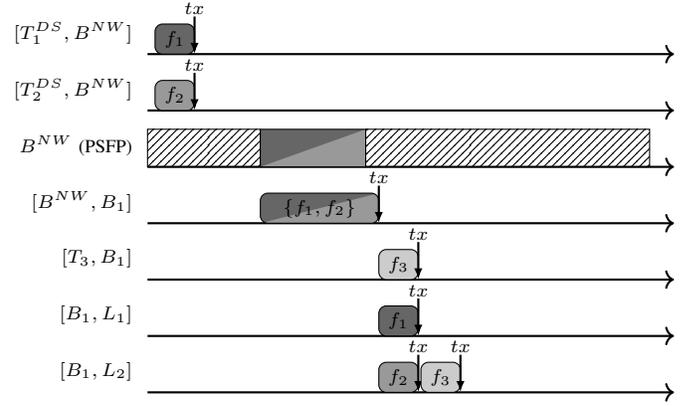

\textbf{C1) Sequential Transmissions.}
For each frame $f \in B_i$, the transmission start $\S{[u,v], B_i}$ is deferred until the latest arrival of $f$ at the bridge $u$, i.e.,
\begin{equation*}
  \S{[u,v], B_i} \geq \Rmax{u, f}.
\end{equation*}
In Fig.~\ref{fig:fips}, the transmission of $\{f_1, f_2\}$ via $[B^{NW}, B_1]$ has to wait until after the latest possible arrival of $f_1$ and $f_2$.
In contrast, for the first transmission hop of a frame, $\Rmax{T_1^{DS}, f_1}$ is set to the earliest release time of $f_1$ by the talker.

\textbf{C2) Transmission Ordering.}
If $B_i$ is not the first batch (within the hypercycle) to be transmitted over $[u,v]$, its transmission is deferred until $B_{i-1}$ is fully transmitted, i.e.,
\begin{equation*}
  \S{[u,v], B_i} \geq \add{\S}{\dpdbmax}{[u,v], B_{i-1}}.
\end{equation*}
For example, the transmission ordering $f_2 \prec f_3$ implies in Fig.~\ref{fig:fips} that the transmission of $f_3$ via $[B_1, L_2]$ can only start after the transmission of $f_2$ has finished.
 
\textbf{C3) Batch Fault Isolation.}
For each frame $f \in B_i$, it must be ensured that $f$ takes its intended transmission slot over the subsequent hop $[v,w]$.
Formally, let $B_j'$ denote the frame batch of $f$ at $[v,w]$.
To ensure $f$ never takes the slot of $B_{j-1}'$, the transmission start of $\S{[u,v], B_i}$ is delayed so that $f$ never arrives at $v$ before the transmission of $B_{j-1}'$ has finished, i.e.,
\begin{equation*}
  \S{[u,v], B_i} \geq \add{\S}{\dpdbmax}{[v,w], B_{j-1}'} - \dpdbmin{[u,v], f}.
\end{equation*}
For example, Fig.~\ref{fig:fips} delays $\S{[T_3, B_1], f_3}$ to ensure $f_3$ never takes $f_2$'s transmission slot over the subsequent link, even if $f_2$ is discarded by PSFP for violating its 5G PDB.
Also note that, in a multi-queue setting, C3 is only necessary if $f_2$ and $f_3$ share the same queue at $[B_1, L_2]$.
Moreover, C3 directly implies that frames are transmitted in a FIFO manner.
 
\textbf{Robust TSN Configuration.}
C1--C3 define a natural recursion order where the transmission start $\S{[u,v], B_i}$ depends on, at most, three earlier transmissions.
In case a recursion cycle is detected, the procedure is aborted;
otherwise, the GCL at link $[u,v]$ is derived as $\GCL([u,v]) = \{[o_i, c_i]\}_{i=1}^n$ with
\begin{equation*}
  o_i = \S{[u,v], B_i} \quad \text{and} \quad c_i = \add{\S}{\dpdbmax}{[u,v], B_i}.
\end{equation*}
Similarly, the PSFP configuration $\R{}$ is set for each frame $f$ and each bridge $v \in \path{f}$ with
\begin{align*}
  \Rmin{v, f} &= \S{[u,v], B_i} + \dpdbmin{[u,v], f}, \quad \text{and} \\
  \Rmax{v, f} &= \add{\S}{\dpdbmax}{[u,v], B_i},
\end{align*}
which is the earliest and latest arrival of $f$ at $v$ under FIPS.

With $\Conf = (\GCL, \R{})$, it is straightforward to check if the latency~(\ref{eq:latency}) and jitter~(\ref{eq:jitter}) requirements are satisfied for all streams. 
As detailed before, FIPS therefore accepts a new stream $F$ if and only if $\Conf$ satisfies the latency and jitter constraints of $F$ and all previously accepted streams.
Finally, we show that FIPS is provably robust (Theorem~\ref{theorem:fips_robust}) and thereby provides formal end-to-end QoS guarantees (Corollary~\ref{corollary:fips_feasibility}).

\begin{theorem} \label{theorem:fips_robust}
  Every TSN configuration $\Conf = (\GCL, \R{})$ derived by FIPS robustly schedules all TSN streams $F \in \TT$.
\end{theorem}
\begin{proof}
  The proof is given in Appendix~\ref{appendix:theorem_proof}.
\end{proof}

\begin{corollary} \label{corollary:fips_feasibility}
  Every TSN configuration $\Conf = (\GCL, \R{})$ derived by FIPS feasibly schedules all accepted TSN streams.
\end{corollary}
\begin{proof}
  By Theorem~\ref{theorem:feasibility}, feasibility is induced by \textit{(i)}--\textit{(iii)}.
  FIPS satisfies
  \begin{enumerate*}[label=(\roman*)]
    \item with Theorem~\ref{theorem:fips_robust}, 
    \item by allocating 5G PDBs as detailed in Section~\ref{sec:robustness:pdb}, and
    \item by accepting only streams for which $\Conf$ satisfies the latency and jitter requirements.
  \end{enumerate*}
\end{proof}

\section{Evaluation} \label{sec:eval}
We compare FIPS against the two groups that existing IEEE 802.1Qbv scheduling techniques fall into: \textit{non-robust} and \textit{strict transmission isolation} approaches.
Section~\ref{sec:eval:non-robust} shows that, compared to non-robust approaches, FIPS can serve a high-criticality stream with a $\qty{99.99}{\percent}$ reliability requirement.
Section~\ref{sec:eval:sti} shows that strict transmission isolation leads to poor scalability, whereas FIPS improves the number of schedulable wireless streams by up to $\times 45$.

\input{eval_topo.tex}

\subsection{Methodology}
We simulate an automated guided vehicle~(AGV) as part of a networked control system.
To this end, we use the OMNeT++ simulation framework~\cite{Varga2010}, including the INET~\cite{Mros2019} and 6GDetCom~\cite{lucas_2023_10401977} extensions.
As shown in Fig.~\ref{fig:eval_topology}, the network is partitioned into the AGV-internal network and the TSN backbone.
Internal TSN devices and switches are connected by $\qty{100}{\mega\bit\per\s}$ Ethernet links with a propagation delay of $\qty{50}{\ns}$. 
The partitions are interconnected by a logical 5G bridge with uplink (DS-TT to NW-TT) and downlink (NW-TT to DS-TT) histograms taken from real 5G measurements \cite{downlink_example_histogram}.

We differentiate between wired and wireless traffic. 
Wired traffic stays within the AGV-internal network or the TSN backbone, with specifications $\size{F} = \qty{100}{\byte}$, $\period{F} = \qty{5}{\ms}$, $\ete{f} = \qty{500}{\us}$, and $\jitter{f} = \qty{1}{\us}$.
Wireless traffic traverses the logical 5G bridge in uplink or downlink direction, with specifications $\size{F} = \qty{100}{\byte}$ and $\period{F} = \qty{20}{\ms}$.
Depending on the evaluated scenario, we consider different QoS requirements for wireless streams.

All computations (simulation and scheduling) are performed on the same machine, equipped with two AMD EPYC 7413 @\qty{3.6}{\giga\hertz} ($2\times24$ cores) and with $\qty{256}{\giga\byte}$ of memory.
\textit{To facilitate reproducibility, we will publish the source code and Docker images for each benchmark once the paper is accepted.}

\subsection{Comparing FIPS with Non-Robust Approaches} \label{sec:eval:non-robust}
To illustrate the importance of robust end-to-end scheduling, we compare the achieved QoS of FIPS and non-robust scheduling approaches.
As existing work typically considers scalar 5G packet delays~\cite{9212049,9940254}, we compare FIPS to using the median (MED) or the maximum (MAX) 5G packet delays from the histograms.

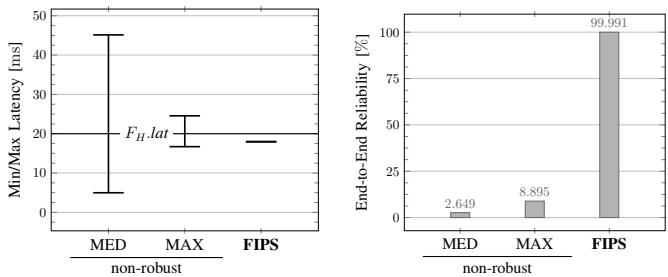
\begin{figure}
  \centering
  \begin{subfigure}{0.48\columnwidth}
    \centering
      \resizebox{\textwidth}{!}{%
      \begin{tikzpicture}
	\begin{axis}[
	      ylabel = {\large Min/Max Latency [\unit{\ms}]},
	      xtick = {1,2,3},
	      xtick align=center, 
	      minor y tick num = 3,
	      enlarge y limits,
	      xmin = 0.25, xmax = 3.75,
	      ymin = 0, ymax = 47,
	      ymajorgrids,
	      xticklabels = {\large MED, \large MAX, \large \textbf{FIPS}},
	      clip=false,
	    ]

	    \draw[thick] (axis cs:0.25,20) -- (axis cs:3.75,20);
	    \node[fill=white] at (axis cs: 1.5,20) {\large $\ete{F_H}$};

	    \intervalplt at (1, 4.983, 45.136);
	    \intervalplt at (2, 16.708, 24.55);
	    \intervalplt at (3, 17.918, 18.017);

	    \draw[thick] (axis cs: 0.5,-11.5) -- (axis cs: 2.5,-11.5);
	    \node at (axis cs: 1.5,-14) {\large non-robust};
	\end{axis}
      \end{tikzpicture}
      }
  \end{subfigure}
  \hfill
  \begin{subfigure}{0.48\columnwidth}
    \centering
      \resizebox{\textwidth}{!}{%
	\begin{tikzpicture}
	  \begin{axis}[
		ybar,
		ylabel = {\large End-to-End Reliability [\unit{\percent}]},
		ytick = {0, 25, 50, 75, 100},
		xtick align=center, 
		minor y tick num = 3,
		ymajorgrids,
		enlarge y limits,
		xtick = {1,2,3},
		xmin = 0.25, xmax = 3.75,
		xticklabels = {\large MED, \large MAX, \large \textbf{FIPS}},
		nodes near coords,
		nodes near coords align={vertical},
		nodes near coords style={/pgf/number format/.cd,precision=3},
		clip=false,
	    ]

	      \addplot[black!60,fill=black!30,bar width=0.5cm] coordinates {(1, 2.6491) (2, 8.8948) (3, 99.9914)};

	      \draw[thick] (axis cs: 0.5,-21) -- (axis cs: 2.5,-21);
	      \node at (axis cs: 1.5,-27) {\large non-robust};
	  \end{axis}
	\end{tikzpicture}
      }
  \end{subfigure}
  \caption{Simulation results showing the achieved QoS of $F_H$.} \label{fig:simulation}
\end{figure}

We analyze the behavior of ten high-criticality wireless streams $F_H$ (five per uplink/downlink direction) with requirements $\ete{F_H} = \qty{20}{\ms}$, $\jitter{F_H} = \qty{100}{\us}$, and $\rel{F_H} = \qty{99.99}{\percent}$.
To increase link congestion, we schedule an additional 10~wired streams (five per wired partition) and 80~wireless streams with $\rel{F} = \qty{50}{\percent}$.
The experiments are repeated for one million hypercycles, a simulation time of approx. $\qty{5.6}{h}$.

Fig.~\ref{fig:simulation} shows the achieved QoS guarantees of a representative high-criticality stream $F_H$.
The results clearly show that the observed end-to-end reliability diminishes for both non-robust approaches below $\qty{10}{\percent}$ and is thereby far from the required $\qty{99.99}{\percent}$.
Moreover, the latency results demonstrate that a frame reordering events as in Section~\ref{sec:problem_description} can create a queueing backlog that pushes frames of $F_H$ until after their deadline or even into subsequent hypercycles.
These results therefore underline the need for provable per-stream QoS guarantees, as we provide with FIPS.

\begin{figure}[b]
  \centering
  \begin{subfigure}{0.48\columnwidth}
    \centering
    \resizebox{\textwidth}{!}{%
    \begin{tikzpicture}
	  \begin{axis}[
		ybar,
		ylabel = {\# Wireless Streams},
		xlabel = {Stream Reliability $\rel{F}$ [$\unit{\percent}$]},
		xtick align=center, 
		minor y tick num = 3,
		ymajorgrids,
		xmin=0.5, xmax=4.5,
		ymax=275,
		xtick = {1,2,3,4},
		xticklabels = {\qty{90}{\percent}, \qty{99}{\percent}, \qty{99.9}{\percent}, \qty{99.99}{\percent}},
		nodes near coords,
		x=1.4cm,
		y=0.016cm,
		nodes near coords align={vertical},
		legend style={at={(0.5,1.05)}, column sep=1ex, anchor=north},
		legend columns=-1,
	    ]
	      \addplot[black!40,fill=black!40,bar width=0.5cm] coordinates {(1, 5) (2, 4) (3, 4) (4, 4)};
	      \addlegendentry{STI}
	      \addplot[black!70,fill=black!70,bar width=0.5cm] coordinates {(1, 226) (2, 180) (3, 143) (4, 143)};
	      \addlegendentry{\textbf{FIPS}}
	  \end{axis}
    \end{tikzpicture}
    }
  \end{subfigure}
  \hfill
  \begin{subfigure}{0.48\columnwidth}
    \centering
    \resizebox{\textwidth}{!}{%
    \begin{tikzpicture}
	  \begin{axis}[
		ylabel = {\# Wireless Streams},
		xlabel = {Jitter $\jitter{F}$ [$\unit{\us}$]},
		xtick align=center, 
		minor y tick num = 3,
		ymajorgrids,
		xmin=-10, xmax=110,
		ymax=275,
		y=0.016cm,
		x=0.045cm,
		legend style={at={(0.44,1.05)}, column sep=1ex, anchor=north},
		legend columns=2,
	    ]
	      \addlegendimage{empty legend}
	      \addlegendentry{\hspace*{-0.75cm}STI}
	      \addlegendimage{empty legend}
	      \addlegendentry{\hspace*{-0.75cm}\textbf{FIPS}}
	      \addplot[black!40,dashed,mark options={black!40,scale=1.2,solid},mark=*] coordinates {(1, 5) (20, 5) (40, 5) (60, 5) (80, 5) (100, 5)};
	      \addlegendentry{$\qty{90}{\percent}$}
	      \addplot[black!70,mark options={black!70,scale=1.2},mark=*] coordinates {(1, 29) (20, 63) (40, 99) (60, 136) (80, 174) (100, 227)};
	      \addlegendentry{$\qty{90}{\percent}$}
	      \addplot[black!40,dashed,mark options={black!40,solid,scale=1.75},mark=triangle*] coordinates {(1, 4) (20, 4) (40, 4) (60, 4) (80, 4) (100, 4)};
	      \addlegendentry{$\qty{99.99}{\percent}$}
	      \addplot[black!70,mark options={black!70,scale=1.75},mark=triangle*] coordinates {(1, 12) (20, 36) (40, 60) (60, 84) (80, 108) (100, 143)};
	      \addlegendentry{$\qty{99.99}{\percent}$}
	  \end{axis}
    \end{tikzpicture}
    }
  \end{subfigure}
  \caption{Scalability results for STI and FIPS. 
    The results on the left are obtained with $\jitter{F} = \qty{100}{\us}$.
  } \label{fig:scalability}
\end{figure}
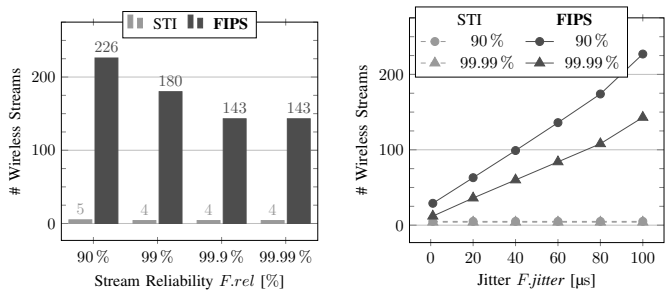

\subsection{Comparing FIPS with Strict Transmission Isolation} \label{sec:eval:sti}
Next, we analyze the impact of robustness on scalability in the number of schedulable streams.
We compare FIPS to strict transmission isolation (STI) approaches, e.g., as commonly employed in wired TSN~\cite{nwps,Craciunas2016RTNS}.
To constitute a fair comparison, STI configurations are derived using the incremental heuristic of FIPS but restrict batching to one frame per batch.

Fig.~\ref{fig:scalability} shows the maximum number of wireless streams for which STI and FIPS find a feasible TSN configuration.
We consider stream sets that consist of 30 internal streams (15 per wired partition) and 400 wireless streams (200 per up-/downlink direction) with randomly generated paths.
The results show the total number of accepted wireless streams, averaged over 100 stream sets, in dependence of different reliability ($\qty{90}{\percent}$ -- $\qty{99.99}{\percent}$) and jitter ($\qty{1}{\us}$ -- $\qty{100}{\us}$) requirements.

Fig.~\ref{fig:scalability} shows the poor scalability of STI with little to no variation over the different QoS requirements.
This is due to the large 5G packet delay variations that cause STI to reserve egress queues at the DS-TT/NW-TT exclusively for individual frames.
In contrast, FIPS shows an expected downward trend for increasing reliability requirements and an upward trend for increasing jitter allowance. 
Hence, we identify the jitter allowance of streams, which essentially restricts the maximum batch size, as the limiting factor for deploying FIPS at scale.

\section{Conclusion} \label{sec:conclusion}
We presented FIPS, a wireless-friendly IEEE 802.1Qbv scheduler that provides formal end-to-end QoS guarantees for each stream.
With the high susceptibility of IEEE 802.1Qbv schedulers to runtime uncertainty, FIPS overcomes the major challenges of \textit{staying robust} against 5G packet delay variations and \textit{remaining scalable} even in 5G bottleneck scenarios.

In future work, we plan to extend the coordination between FIPS and the 5G system, allowing for more fine-grained 5G QoS provisioning that is captured by the 5G PDB contracts.

\section*{Acknowledgment}
This work was supported by the European Commission through the H2020 project DETERMINISTIC6G (Grant Agreement no. 101096504).
We thank Melanie Heck and the DETERMINISTIC6G consortium for their valuable feedback.

\bibliographystyle{IEEEtran.bst}
\bibliography{IEEEabrv,bibliography.bib}

\clearpage
\onecolumn

\appendix
\begin{figure*}
\centering
    \begin{subfigure}{0.1\textwidth}
	\begin{tikzpicture}
	    \begin{scope}[scale=0.75]
		\node at (0,0.5) {\footnotesize $\T{[u^N,v]}$};
		\node at (0,1.5) {\footnotesize $\R{u^N}$};
		\node at (0,2.5) {\footnotesize $\T{[u_2^D, u^N]}$};
		\node at (0,3.5) {\footnotesize $\T{[u_1^D, u^N]}$};
		\draw[white] (0,4)--(1,4);
		\draw[white] (0,-0.8)--(1,-0.8);
	    \end{scope}
	\end{tikzpicture}
    \end{subfigure}
    \begin{subfigure}{0.25\textwidth}
	\centering
	\begin{tikzpicture}
	    \begin{scope}[scale=0.75]
		\draw[thick] (0,0) -- (5,0);
		\draw[thick] (0,1) -- (5,1);
		\draw[thick] (0,2) -- (5,2);
		\draw[thick] (0,3) -- (5,3);
		\draw[thick] (0,4) -- (5,4);

		\draw[pattern=north east lines] (0.5,3) rectangle ++(2,0.7);
		\draw[fill=black!20] (1,3) rectangle ++(1.3,0.6) node[pos=.5] {$f_1$};
		\draw[<-, thick] (0.5,3.9) -- (0.5,3);

		\draw[pattern=north west lines] (1,2) rectangle ++(2,0.7);
		\draw[fill=black!20] (1,2) rectangle ++(1.8,0.6) node[pos=.5] {$f_2$};
		\draw[<-, thick] (1,2.9) -- (1,2);

		\draw[pattern=north east lines] (1.5,1) rectangle ++(1,0.7);
		\draw[pattern=north west lines] (2,1) rectangle ++(1,0.7);

		\draw[pattern=north west lines] (2.5,0) rectangle ++(2,0.7);
		\draw[pattern=north east lines] (2.5,0) rectangle ++(2,0.7);
		\draw[fill=black!20] (2.5,0) rectangle ++(1,0.6) node[pos=.5] {$f_2$};
		\draw[fill=black!20] (3.2,0) rectangle ++(1,0.6) node[pos=.5] {$f_1$};
	    \end{scope}
	\end{tikzpicture}
      \caption{Basic Consistency} \label{fig:exec_seq:basic_consistency}
    \end{subfigure}
    \begin{subfigure}{0.25\textwidth}
	\centering
	\begin{tikzpicture}
	    \begin{scope}[scale=0.75]
		\draw[thick] (0,0) -- (5,0);
		\draw[thick] (0,1) -- (5,1);
		\draw[thick] (0,2) -- (5,2);
		\draw[thick] (0,3) -- (5,3);
		\draw[thick] (0,4) -- (5,4);

		\draw[pattern=north east lines] (0.5,3) rectangle ++(2,0.7);
		\draw[fill=black!20] (0.5,3) rectangle ++(1.3,0.6) node[pos=.5] {$f_1$};
		\draw[<-, thick] (0.5,3.9) -- (0.5,3);

		\draw[pattern=north west lines] (1,2) rectangle ++(2,0.7);
		\draw[fill=black!20] (1,2) rectangle ++(1.8,0.6) node[pos=.5] {$f_2$};
		\draw[<-, thick] (1,2.9) -- (1,2);

		\draw[pattern=north east lines] (1.5,1) rectangle ++(1,0.7);
		\draw[pattern=north west lines] (2,1) rectangle ++(1,0.7);

		\draw[pattern=north west lines] (2.5,0) rectangle ++(2,0.7);
		\draw[pattern=north east lines] (2.5,0) rectangle ++(2,0.7);
		\draw[fill=black!20] (2,0) rectangle ++(1,0.6) node[pos=.5] {$f_1$};
		\draw[fill=black!20] (3.5,0) rectangle ++(1,0.6) node[pos=.5] {$f_2$};
	    \end{scope}
	\end{tikzpicture}
	\caption{GCL Consistency} \label{fig:exec_seq:gcl_consistency}
    \end{subfigure}
    \begin{subfigure}{0.25\textwidth}
	\centering
	\begin{tikzpicture}
	    \begin{scope}[scale=0.75]
		\draw[thick] (0,0) -- (5,0);
		\draw[thick] (0,1) -- (5,1);
		\draw[thick] (0,2) -- (5,2);
		\draw[thick] (0,3) -- (5,3);
		\draw[thick] (0,4) -- (5,4);

		\draw[pattern=north east lines] (0.5,3) rectangle ++(2,0.7);
		\draw[fill=black!20] (0.5,3) rectangle ++(1.5,0.6) node[pos=.5] {$f_1$};
		\draw[<-, thick] (0.5,3.9) -- (0.5,3);

		\draw[pattern=north west lines] (1,2) rectangle ++(2,0.7);
		\draw[fill=black!20] (1,2) rectangle ++(2,0.6) node[pos=.5] {$f_2$};
		\draw[<-, thick] (1,2.9) -- (1,2);

		\err at (3,2);

		\draw[pattern=north east lines] (1.5,1) rectangle ++(1,0.7);
		\draw[pattern=north west lines] (2,1) rectangle ++(1,0.7);

		\draw[pattern=north west lines] (2.5,0) rectangle ++(2,0.7);
		\draw[pattern=north east lines] (2.5,0) rectangle ++(2,0.7);
		\draw[fill=black!20] (3,0) rectangle ++(1,0.6) node[pos=.5] {$f_2$};
	    \end{scope}
	\end{tikzpicture}
	\caption{Stream Policing} \label{fig:exec_seq:stream_policing}
    \end{subfigure}
    \caption{Invalid execution sequences violating Basic Consistency (a), GCL Consistency (b), or Stream Policing (c)}
    \label{fig:exec_seq}
\end{figure*}
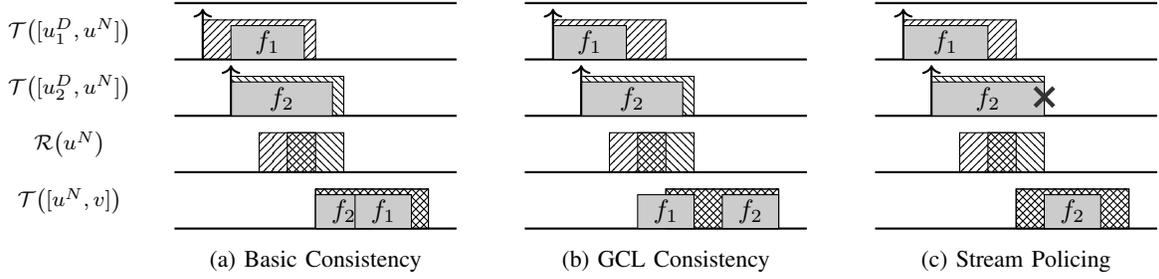

\subsection{Formal Definition of Execution Sequences}
For a network graph $G$ with TSN configuration $\Conf = (\R{}, \GCL{})$, we define an execution sequence $\E(G, \C)$ as a tuple $(\T{}, \D{})$,\footnote{If both $G$ and $\Conf$ are clear from context, we simply write $\E$.} where, for each link $[u,v] \in E$ and frame $f \in \TT_{[u,v]}$
\begin{itemize}
    \item $\T{[u,v],f}$ describes the transmission offset in $\E$, and
    \item $\D{[u,v],f}$ describes the transmission delay in $\E$.
\end{itemize}
To model dropped frames, an execution sequence is allowed to assign $\infty$ to $\T{}$ and $\D{}$.

Formalizing execution sequences is motivated by two aspects.
First, formal constraints are required to prove robustness of a TSN configuration $\Conf$.
Second, formal constraints provide valuable insight into the technical requirements the network infrastructure must provide to guarantee the services' reliability requirements.
We start in Section~\ref{sec:execution_sequences:notation} with notational conventions.
Thereafter, the constraints \textit{Basic Consistency}, \textit{GCL Consistency}, and \textit{Stream Policing} are introduced.
To illustrate the effects and the importance of each constraint, we refer to Fig.~\ref{fig:exec_seq} throughout these sections.

\subsubsection{Notational Conventions}\label{sec:execution_sequences:notation}
An execution sequence $\E$ specifies the transmission offsets $\T{}$ and the transmission delays $\D{}$. 
To model dropped frames of a TSN stream $f \in \TT$, we allow $\T{}$ and $\D{}$ to map to $\infty$, conveying the following semantics:
\begin{enumerate}[label=I\arabic*)]
  \item If $\T{[\v{k}{f}, \v{k+1}{f}],f} < \infty$ but $\D{[\v{k}{f}, \v{k+1}{f}],f} = \infty$, $f$ is dropped by $\v{k}{f}$ during its transmission over $[\v{k}{f}, \v{k+1}{f}]$.
  \item If $\D{[\v{k}{f}, \v{k+1}{f}],f} < \infty$ but $\T{[\v{k+1}{f}, \v{k+2}{f}],f} = \infty$, $f$ is dropped by $\v{k+1}{f}$ upon receipt.
\end{enumerate}

While $\D{[\v{k}{f}, \v{k+1}{f}],f} = \infty$ and I1 capture the message streams' view, we also require a formalism that tells $\v{k}{f}$ when it is allowed to start the transmission of the subsequent frame.\footnote{In fluid dynamics, these two points of view correspond to the Lagrangian and Eulerian observers, respectively.}
For this purpose, we define the \textit{effective} transmission delay $\eD{}$ to be equal to the upper PDB bound $\dpdbmax{}$, if I1 occurs, i.e.,
\begin{equation}\label{eq:eD}
  \eD{[\v{k}{f}, \v{k+1}{f}],f} = \begin{cases}
    \dpdbmax{[\v{k}{f}, \v{k+1}{f}],f} \quad &\text{for } \text{I1}, \\
    \D{[\v{k}{f}, \v{k+1}{f}],f} \quad &\text{else.}
    \end{cases}
\end{equation}

\subsubsection{Basic Consistency} \label{sec:execution_sequences:basic_consistency}
Basic Consistency restrains the execution sequence $\E$ to ensure mutual exclusion of transmission links (i.e., Ethernet links and 5G frequency channels), sequential transmission, isochronous talkers, and FIFO queuing.

\textbf{\TC} \label{sec:execution_sequences:basic_consistency:tc}
ensures that for every link $[u,v] \in E$, no two frames $f, f' \in \TT_{[u,v]}$ with $f \neq f'$ transmit at the same time, i.e.,
\begin{align*}
  \add{\T}{\eDtrans}{[u,v], f} &\leq \T{[u,v], f'}, \qquad \text{or}\\
  \add{\T}{\eDtrans}{[u,v], f'} &\leq \T{[u,v], f}.
\end{align*}
Fig.~\ref{fig:exec_seq:basic_consistency} violates this constraint by overlapping the transmission of frames $f_1$ and $f_2$ at $[u^N, v]$.

\textbf{\ST} \label{sec:execution_sequences:basic_consistency:st}
ensures that for every stream $f \in \TT$ and every $1 \leq k < n(f)$ that the transmission of $f$ over $[\v{k}{f}, \v{k+1}{f}]$ can only start after $f$ is fully received by $\v{k}{f}$, i.e.,
\begin{equation*}
  \add{\T}{\D}{[\v{k}{f}, \v{k+1}{f}], f} \leq \T{[\v{k+1}{f}, \v{k+2}{f}], f}.
\end{equation*}
Fig.~\ref{fig:exec_seq:basic_consistency} violates this constraint by starting the transmission of $f_2$ at $[u^N, v]$ before $f_2$'s transmission over $[u^D_2, u^N]$ completed.

\textbf{$\IT$} \label{sec:execution_sequences:basic_consistency:it}
ensures that for every stream $f \in \TT$ that the talker $\v{1}{f}$ starts its transmission as configured 
\begin{equation*}
  \T{[\v{1}{f}, \v{2}{f}], f} = \Smin{[\v{1}{f}, \v{2}{f}], f}.
\end{equation*}
Fig.~\ref{fig:exec_seq:basic_consistency} violates this constraint by starting the transmission of $f_1$ at $[u_1^D, u^N]$ too late.

\textbf{$\FIFO$} \label{sec:execution_sequences:basic_consistency:fifo}
ensures that for every link $[v,w] \in E$ and every two frames $f, f' \in \TT_{[v,w]}$, with respective previous hops $[u,v]$ and $[u',v]$, that the transmission order is identical to the receiving order, i.e,
\begin{align*}
  \add{\T}{\D}{[u,v], f} < \add{\T}{\D}{[u',v], f'} \\
    \iff \T{[v,w], f} < \T{[v,w], f'}.
\end{align*}
Fig.~\ref{fig:exec_seq:basic_consistency} violates this constraint by transmitting $f_2$ over $[u^N, v]$ before $f_1$, whereas $f_1$ arrives at $u^N$ before $f_2$.

\subsubsection{GCL Consistency}\label{sec:execution_sequences:gcl}
GCL Consistency restrains the execution sequence $\E$ to respect the configured GCL $\GCL([v,w])$ for every link $[v,w] \in E$.
For this purpose, let
\begin{equation*}
    \GCL([v,w]) = [(o_1,c_1), \ldots, (o_m,c_m)].
\end{equation*}
GCL Consistency consists of encapsulation and progress.

\textbf{$\GCLC$} \label{sec:execution_sequence:gcl:completeness}
ensures for every arrived frame $f \in \TT_{[v,w]}$, i.e., $\T{[v,w],f} < \infty$, that $v$ transmits $f$ within a configured window $(o_i, c_i)$, i.e.,
\begin{equation}\label{eq:gclc}
    o_i \leq \T{[v,w], f} \leq c_i - \dpdbtransmax{[v,w], f}.
\end{equation}
Fig.~\ref{fig:exec_seq:gcl_consistency} violates this constraint by starting the transmission of $f_1$ at $[u^N, v]$ before the gate opens.

\textbf{$\GCLS$}\label{sec:execution_sequence:gcl:soundness}
ensures for every transmission window $(o_i, c_i)$ and every point in time $o_i \leq t < c_i$ that, if the queue at $[v,w]$ is not empty, i.e., there exists a frame $f \in \TT_{[v,w]}$ with previous hop $[u,v]$ and
\begin{equation*}
  \add{\T}{\D}{[u,v], f} \leq t \leq \T{[v,w], f} < \infty, 
\end{equation*}
there exists $f' \in \TT_{[v,w]}$, potentially different from $f$, which is transmitting over $[v,w]$ at time $t$, i.e.,
\begin{equation*}
  \T{[v,w], f'} \leq t < \add{\T}{\eDtrans}{[v,w], f'}. 
\end{equation*}
Fig.~\ref{fig:exec_seq:gcl_consistency} violates this constraint by having a gap between $f_1$'s and $f_2$'s transmission at $[u^N, v]$, although the gate remains open and $f_2$ arrives at $u^N$ before $f_1$'s transmission finishes.

\subsubsection{Stream Policing} \label{sec:execution_sequences:stream_policing}
Stream Policing restrains the execution sequence $\E$ to drop frames correctly. 
There are two important aspects that have to be covered to correctly model wireless TSN behavior.
First, a frame is dropped by a TSN bridge if it
\begin{enumerate}[label=I\arabic*)]
  \item transmits longer than the PDB $\dpdb{}$, or
    \item arrives outside the PSFP interval.
\end{enumerate}
Second, $\E$ cannot drop frames arbitrarily, that is, if a frame is dropped, it must be because of I1 or I2.

\textbf{$\TP$}
ensures for every stream $f \in \TT$ that $f$ is discarded during transmission over $[\v{k}{f}, \v{k+1}{f}]$ if and only if the transmission duration takes longer than the PDB, i.e.,
\begin{align*}
  \Dtrans{[\v{k}{f}, \v{k+1}{f}],f} > \dpdbtransmax{[\v{k}{f}, \v{k+1}{f}],f} \\
    \iff \Dtrans{\clink{k}{f},f} = \infty.
\end{align*}
Fig.~\ref{fig:exec_seq:stream_policing} shows this behavior by dropping $f_2$ at $[u_2^D, u^N]$, i.e., $\Dtrans{[u_2^D, u^N], f_2} = \infty$, as indicated by the cross, after its transmission is not successfully completed within its PDB.

\textbf{$\PSFP$}
ensures for every stream $f \in \TT$ that $f$ is dropped by PSFP at $\v{k}{f}$ ($k > 1$), if and only if it arrives outside the PSPF interval, i.e.,
\begin{align*}
  \add{\T}{\D}{\pclink{k}{f},f} \notin \R{\v{k}{f},f} \\
    \iff \T{\clink{k}{f},f} = \infty.
\end{align*}
Fig.~\ref{fig:exec_seq:stream_policing} violates this constraint by dropping $f_1$ at $u^N$, i.e., $\T{[u^N, v], f_1} = \infty$, although it arrived within the PSFP interval.

\textbf{$\PC$}
is used for consistent modelling of the execution sequence $\E$, ensuring that once a frame is dropped, it remains dropped for all subsequent hops.
Thus, we require for every stream $f \in \TT$ and every hop $\clink{k}{f}$ that
\begin{equation} \label{eq:pc:1}
  \T{\clink{k}{f},f} = \infty \implies \D{\clink{k}{f},f} = \infty.
\end{equation}
Further, we require $\clink{k}{f}$ with $k < n - 1$ to satisfy
\begin{equation} \label{eq:pc:2}
  \D{\clink{k}{f},f} = \infty \implies \T{\nclink{k}{f},f} = \infty.
\end{equation}
While (\ref{eq:pc:2}) would also be satisfied by combining (\ref{eq:pc:1}) with $\textsf{PSFP}$ or $\textsf{SequentialTransmission}$, we prefer an explicit separation of constraint functionality.
Fig.~\ref{fig:exec_seq:stream_policing} violates this constraint by transmitting $f_2$ over $[u^N, v]$, i.e., $\T{[u^N, v], f_2} < \infty$, although it was previously dropped by $\TP$ during its transmission at $[u_2^D, u^N]$.

\subsection{Proof of Theorem~\ref{theorem:feasibility}} \label{appendix:feasibility}
By allocating sufficiently large 5G PDBs to satisfy~(\ref{eq:pdb2}) for stream $F$, the following holds true for all frames $f \in F$:
\begin{align}
  &\rel{F} \leq \prod_{[u,v] \in \path{F}} \Prob\big(\D{[u,v], f} \in \dpdb{[u,v], f}\big) \\
  = &\sum_{\E} \prod_{[u,v]} \Prob\big(\D{[u,v], f} \in \dpdb{[u,v], f} \mid \E \big) \times \Prob(\E), \label{eq:ap1}
\end{align}
By Definition~\ref{def:robustness}, the arrival time $A(f)$ is within the expected interval $\R{\vl{F}, f}$ if all packet delays along $\path{F}$ lie within their respective PDBs.
We can write 
\begin{align*}
  (\ref{eq:ap1}) &\leq \sum_{\E} \Prob\big(A(f) \in \R{\vl{F}, f} \mid \E\big) \times \Prob(\E) \\
   &= \Prob\big(A(f) \in \R{\vl{F}, f}\big).
\end{align*}
Hence, the probability of $f$ arriving within $\R{\vl{F}, f}$ exceeds the required $\rel{F}$.
Also, (iii) of Theorem~\ref{theorem:feasibility} constrains $\R{\vl{F}, f}$ to satisfy the latency and jitter bounds of $F$.

\subsection{Proof of Theorem~\ref{theorem:fips_robust}} \label{appendix:theorem_proof}
We start with some notational conventions:
First, compared to $\S{[u,v], B_i}$, which defines the transmission start of a frame batch at $[u,v]$, we want to verify when a frame $f \in B_i$ can start its transmission.
We therefore define the interval $\S{} = [\Smin{}, \Smax{}]$ to capture the (intended) minimum and maximum transmission times, i.e.,
\begin{align} \label{eq:sminmax}
  \Smin{[u,v], f} = \S{[u,v], B_i} \quad \text{and} \quad \Smax{[u,v], f} = \sum_{f' \in B_i \setminus \{f\}} \dpdbtransmax{[u,v], f'},
\end{align}
for each $[u,v] \in E$ and $f \in \TT_{[u,v]}$.
Here, $\dpdbtransmax{}$ solely captures the serialization delay (for both Ethernet links and 5G links).
Moreover, to have a clear mapping between frames $f$ and their corresponding batch, we write $\I{[u,v], f} = B_i$.

Given an execution sequence $\E = (\T{}, \D{})$, we define $\S{}^\E$ as the \textit{expected transmission offset for $\E$} with
\begin{equation} \label{eq:fips:se}
  \S{}^\E([u,v],f) = \Smin{[u,v],f} + \sum_{f' \in B^\E([u,v], f)} \eDtrans{[u,v],f'},
\end{equation}
for $[u,v] \in E$ and $f \in \TT_{[u,v]}$, where $B^\E([u,v], f)$ denotes frames $f' \in B([u,v], f)$ which are transmitted over $[u,v]$ before $f$, i.e., $\T{[u,v],f'} < \T{[u,v],f}$.
We commence with structured proofs to show the robustness of FIPS.

\begin{lemma}\label{lemma:fips:aux}
  Let $\S{}$ denote a schedule that concurs with FIPS, and $\Conf$ the derived FIPS configuration. 
  Then, for every execution sequence $\E = (\T{}, \D{})$ we have $\S{}^\E([u,v],f) \in \S{[u,v],f}$, for all $[u,v] \in E$ and $f \in \TT_{[u,v]}$.
\end{lemma}
\begin{proof}
    \textit{Proof:}
    \pflet{$[u,v] \in E$ and $f \in \TT_{[u,v]}$}
    \step{<2>1}{$\SE{[u,v],f} \geq \Smin{[u,v],f}$}
    \begin{proof}
      \pf\ By definition of $\SE{[u,v],f}$ in (\ref{eq:fips:se}).
    \end{proof}
    \step{<2>2}{$\SE{[u,v],f} \leq \Smax{[u,v],f}$}
    \begin{proof}
      \step{<3>1}{$\IE{[u,v], f} \subseteq \I{[u,v], f} \backslash \{f\}$}
	\begin{proof}
	    \pf\ By definition of $\IE{[u,v], f}$.
	\end{proof}
	\step{<3>2}{$\eDtrans{[u,v],f'} \leq \dpdbtransmax{[u,v],f'}$ for all $f' \in \IE{[u,v],f}$}
	\begin{proof}
	    \step{<4>1}{$\T{[u,v], f'} < \infty$}
	    \begin{proof}
		\pf\ By definition of $\IE{[u,v], f'}$, we have $\T{[u,v], f'} < \T{[u,v], f}$
	    \end{proof}
	    \step{<4>2}{\case{$\Dtrans{[u,v],f'} = \infty$}}
	    \begin{proof}
		\pf\ $\eDtrans{[u,v],f'} = \dpdbtransmax{[u,v], f'}$ by (\ref{eq:eD}) with case assumption \stepref{<4>2} and step \stepref{<4>1}.
	    \end{proof}
	    \step{<4>3}{\case{$\Dtrans{[u,v],f'} < \infty$}}
	    \begin{proof}
		\step{<5>1}{$\eDtrans{[u,v],f'} = \Dtrans{[u,v],f'}$}
		\begin{proof}
		    \pf\ By definition of $\eD{[u,v],f'}$ in (\ref{eq:eD}) with case assumption \stepref{<4>3}.
		\end{proof}
		\step{<5>2}{$\Dtrans{[u,v],f'} \leq \dpdbtransmax{[u,v],f'}$}
		\begin{proof}
		    \pf\ By $\TP$ and case assumption \stepref{<4>3}.
		\end{proof}
		\qedstep
		\begin{proof}
		    \pf\ By steps \stepref{<5>1} and \stepref{<5>2}.
		\end{proof}
	    \end{proof}
	\end{proof}
	\qedstep
	\begin{proof}
	    \pf\ By steps \stepref{<3>1}, \stepref{<3>2}, we have
	    \begin{align*}
	      \SE{[u,v],f} \leq \Smin{[u,v],f} + \sum_{f' \in \I{[u,v], f} \backslash \{f\}} \dpdbtransmax{[u,v],f'} = \Smax{[u,v],f}
	    \end{align*}
	\end{proof}
    \end{proof}
    \qedstep
    \begin{proof}
	\pf\ By steps \stepref{<2>1} and \stepref{<2>2}. 
    \end{proof}
\end{proof}

\begin{lemma}\label{lemma:fips}
  Let $\S{}$ denote a schedule that concurs with FIPS, and $\Conf$ the derived FIPS configuration. 
    Then, for every execution sequence $\E = (\T{}, \D{})$ and every message stream $f \in \TT$, the following holds true:
    For $1 \leq l < n(f)$, if $f$ arrives at $\v{l}{f}$ within its PSFP enforced interval, i.e.,
    \begin{equation*}
      \add{\T}{\D}{\pclink{l}{f},f} \in \R{\v{l}{f},f},
    \end{equation*}
    then $\v{l}{f}$ starts the transmission of $f$ as expected by $\S{}^\E$, i.e.,
    \begin{equation*}
      \T{\clink{l}{f},f} = \SE{\clink{l}{f},f}.
    \end{equation*}
\end{lemma}
\begin{proof}
    \textit{Proof:}
    \define{$\V \subseteq E \times \TT$ with $(\clink{l}{f},f) \in \V$, if and only if
	\begin{pfenum}
	\item $\add{\T}{\D}{\pclink{l}{f},f} \in \R{\v{l}{f}, f}$, and
	\item  $\T{\clink{l}{f}, f} \neq \SE{\clink{l}{f}, f}$.
	\end{pfenum}
    }
    \assume{$\V \neq \emptyset$}
    \prove{False}
    \explan{Proof by contradiction.
	Assume $\V \neq \emptyset$ and pick $(\clink{l}{f},f) \in \V$ which is the ``first'' to deviate from the expected transmission offset (step \estepref{1}{1}).
	There are three possible cases to consider: 
	\begin{itemize}
	  \item Step \estepref{1}{4}: the frame $f$ starts its transmission too late, i.e., after $\SE{\clink{l}{f},f} < \T{\clink{l}{f},f}$,
	  \item Step \estepref{1}{5}: the frame $f$ starts its transmission too soon but still within $\S{\clink{l}{f},f}$, i.e., $\Smin{\clink{l}{f},f} \leq \T{\clink{l}{f},f} < \SE{\clink{l}{f},f}$, or
	  \item Step \estepref{1}{6}: the frame $f$ starts its transmission too soon and even before $\Smin{\clink{l}{f},f}$, i.e., $\T{\clink{l}{f},f} < \Smin{\clink{l}{f},f}$.
	\end{itemize}
	We derive a contradiction for each case.
    }
    \pflet{$\GCL(\clink{l}{f}) = [(o_1, c_1), \ldots, (o_m, c_m)]$}
    \step{<1>1}{\pick{ $(\clink{l}{f},f) \in \V$ with $\T{\clink{l},f}$ minimal}}
    \begin{proof}
        \pf\ Exists by assumption {\toplevel} and since $\V \subseteq E \times \TT$ is finite.
    \end{proof}
    \step{<1>3}{$\T{\clink{l}{f},f} < \infty$}
    \begin{proof}
	\pf\ By $\PSFP$ with step \stepref{<1>1} and definition \toplevel:1.
    \end{proof}
    \step{<1>4}{\pick{ $1 \leq i \leq m$ with $o_i = \Smin{\clink{l}{f},f}$}}
    \begin{proof}
	\pf\ Exists by construction of FIPS's GCL encoding.
    \end{proof}
    \step{<1>5}{\case{$\SE{\clink{l}{f},f} < \T{\clink{l}{f},f}$}}
    \begin{proof}
      \explan{In case the frame $f$ starts its transmission later than expected, we show that there must be some other frame $f'$ that is still transmitting over $\clink{l}{f}$ at time $\SE{\clink{l}{f},f}$ (step \estepref{2}{1}). 
	    We derive that $f'$ must also start its transmission later than expected (step \estepref{2}{4}), as otherwise its transmission would be aborted before $\SE{\clink{l}{f},f}$.
	    This, however, yields a contradiction to $\T{\clink{l}{f},f}$ being minimal (step \estepref{2}{5}).
        }
	\step{<2>1}{\pick{ $f' \in \TT_{\clink{l}{f}}$, with $f \neq f'$, such that $f'$ is still transmitting over $\clink{l}{f}$ at time $\SE{\clink{l},f}$, i.e., 
        \begin{align*}
	  \T{\clink{l}{f}, f'} \leq \SE{\clink{l}{f},f} < \add{\T}{\eDtrans}{\clink{l}{f}, f'}. 
        \end{align*}
        }}
        \begin{proof}
            \explan{
                The existence of $f'$ follows by $\GCLS$. 
                We therefore have to verify that its precondition is fulfilled.
                On the other hand, $f' \neq f$ as otherwise $f$ would not be delayed.
            }
	    \step{<3>1}{The egress queue at $\clink{l}{f}$ is not empty at time $\SE{\clink{l}{f},f}$, i.e.,
            \begin{align*}
	      \add{\T}{\D}{\pclink{l}{f}, f} \leq \SE{\clink{l}{f},f} \leq \T{\clink{l}{f}, f} < \infty
            \end{align*}}
            \begin{proof}
                \explan{We derive the inequality step by step.}
		\step{<4>1}{$\add{\T}{\D}{\pclink{l}{f}, f} \leq \Rmax{\v{l}{f}, f}$} 
                \begin{proof}
                    \pf\ By step \stepref{<1>1} and definition \toplevel:1.
                \end{proof}
		\step{<4>2}{$\Rmax{\v{l}{f}, f} \leq \Smin{\clink{l}{f},f}$} 
                \begin{proof}
                    \pf\ By C1.
                \end{proof}
		\step{<4>3}{$\Smin{\clink{l}{f}, f} \leq \SE{\clink{l}{f},f}$} 
                \begin{proof}
                    \pf\ By Lemma~\ref{lemma:fips:aux}.
                \end{proof}
                \qedstep
                \begin{proof}
                    \pf\ By steps \stepref{<4>1}, \stepref{<4>2}, \stepref{<4>3}, \stepref{<1>5} (case assumption), and \stepref{<1>3}, in that order.
                \end{proof}
            \end{proof}
            \qedstep
            \begin{proof}
                \pf\ $f'$ exists by $\GCLS$ with step \stepref{<3>1}.
		We have $f' \neq f$, since $\T{\clink{l}{f}, f'} \leq \SE{\clink{l}{f}, f} < \T{\clink{l}{f}, f}$ by case assumption \stepref{<1>5}.
            \end{proof}
        \end{proof}
	\pflet{$l'$ with $\v{l}{f} = \v{l'}{f'}$.}
	\step{<2>2}{$\add{\T}{\D}{\pclink{l'}{f'},f'} \in \R{\v{l}{f}, f'}$}
        \begin{proof}
	  \pf\ By $\PSFP$ with $\T{\clink{l}{f},f'} < \infty$ by step \stepref{<2>1}.
        \end{proof}
	\step{<2>3}{$\add{\SE}{\eDtrans}{\clink{l}{f}, f'} \leq \SE{\clink{l}{f}, f}$}
        \begin{proof}
            \explan{The proof of this inequality depends on whether $f$ and $f'$ are allowed to interleave.}
	    \step{<3>1}{\case{$f' \in \I{\clink{l}{f},f}$}}
            \begin{proof}
	      \explan{If $f$ and $f'$ are allowed to interleave, the inequality directly follows from the definition of $\IE{}$ and $\SE{}$.}
		\step{<4>1}{$f' \in \IE{\clink{l}{f},f}$}
        	\begin{proof}
		  \pf\ $\T{\clink{l}{f},f'} < \T{\clink{l}{f},f}$ by step \stepref{<2>1} and case assumption \stepref{<1>4}.
        	\end{proof}
        	\qedstep
        	\begin{proof}
		  \pf\ The inequality \stepref{<2>3} holds by step \stepref{<4>1} and definition of $\SE{}$ in (\ref{eq:fips:se}).
        	\end{proof}
            \end{proof}
	    \step{<3>2}{\case{$f' \notin \I{\clink{l}{f},f}$}}
            \begin{proof}
                \explan{If $f$ and $f'$ are strictly separated, we highly depend on C3 to yield the inequality \stepref{<2>3}.
                We show the inequality step by step.}
		\step{<4>1}{$\add{\SE}{\eDtrans}{\clink{l}{f}, f'} \leq \add{\Smax}{\eDtrans}{\clink{l}{f}, f'}$}
                \begin{proof}
                    \pf\ By Lemma~\ref{lemma:fips:aux}
                \end{proof}
		\step{<4>2}{$\add{\Smax}{\eDtrans}{\clink{l}{f}, f'} \leq \add{\Smax}{\dpdbtransmax}{\clink{l}{f}, f'}$}
                \begin{proof}
		  \explan{We check each case of the effective transmission delay's definition (\ref{eq:eD}), using that $\T{\clink{l}{f},f'} < \infty$ by step \stepref{<2>1}.}
		    \step{<5>1}{\case{$\D{\clink{l}{f},f'} = \infty$}}
                    \begin{proof}
		      \pf\ Then $\eDtrans{\clink{l}{f}, f'} = \dpdbtransmax{\clink{l}{f}, f'}$, since (\ref{eq:eD}) with I1 ($\T{\clink{l}{f},f'} < \infty$ by step \stepref{<2>1}).
                    \end{proof}
		    \step{<5>2}{\case{$\D{\clink{l}{f},f'} < \infty$}}
                    \begin{proof}
                	\pf\ Then 
                	\begin{equation*}
			  \eDtrans{\clink{l}{f}, f'} = \D{\clink{l}{f}, f'} \leq \dpdbtransmax{\clink{l}{f}, f'},  
                	\end{equation*}
                	 by (\ref{eq:eD}) and $\TP$.
                    \end{proof}
                    \qedstep
                    \begin{proof}
                	\pf\ By steps \stepref{<5>1} and \stepref{<5>2}.
                    \end{proof}
                \end{proof}
		\step{<4>3}{$\add{\Smax}{\dpdbtransmax}{\clink{l}{f}, f'} \leq \Rmin{\v{l}{f}, f}$}
                \begin{proof}
                    \explan{We show that C3's preconditions for the strict separation are satisfied.}
		    \step{<5>1}{$\R{\v{l}{f}, f'} \cap \R{\v{l}{f}, f} = \emptyset$} 
                    \begin{proof}
		      \pf\ By definition of $\I{}$ with case assumption \stepref{<3>2}. 
                    \end{proof}
		    \step{<5>2}{$\Rmin{\v{l}{f}, f'} \leq \Rmax{\v{l}{f}, f}$}
                    \begin{proof}
		      \pf\ By step \stepref{<2>1} and case assumption \stepref{<1>4}, we have $\T{\clink{l}{f},f'} \leq \T{\clink{l}{f},f}$.
			By step \stepref{<2>2} and \stepref{<1>1}, both $f'$ and $f$ arrive at $\v{l}{f}$ within their PSFP enforced intervals.
                	Therefore, $\FIFO$ yields
                	\begin{equation*}
			  \Rmin{\v{l}{f}, f'} \leq \add{\T}{\D}{\pclink{l'}{f'},f'} < \add{\T}{\D}{\pclink{l}{f},f} \leq \Rmax{\v{l}{f},f}.
                	\end{equation*}
                    \end{proof}
                    \qedstep
                    \begin{proof}
		      \pf\ Steps \stepref{<5>1} and \stepref{<5>2} yield $\Rmax{\v{l}{f},f'} < \Rmin{\v{l}{f},f}$.
                	Therefore, C3 yields the inequality of \stepref{<4>3}.
                    \end{proof}
                \end{proof}
		\step{<4>4}{$\Rmin{\v{l}{f},f} \leq \SE{\clink{l}{f},f}$}
                \begin{proof}
                    \pf\ By C1 and Lemma~\ref{lemma:fips:aux}.
                \end{proof}
                \qedstep
                \begin{proof}
                    \pf\ The inequality \stepref{<2>3} holds by steps \stepref{<4>1}--\stepref{<4>4}.
                \end{proof}
            \end{proof}
            \qedstep
            \begin{proof}
                \pf\ The inequality \stepref{<2>3} is shown for both cases \stepref{<3>1} and \stepref{<3>2}.
            \end{proof}
        \end{proof}
	\step{<2>4}{$\SE{\clink{l}{f}, f'} < \T{\clink{l}{f}, f'}$}
        \begin{proof}
            \explan{We show the inequality step by step.}
	    \step{<3>1}{$\add{\SE}{\eDtrans}{\clink{l}{f}, f'} \leq \SE{\clink{l}{f}, f}$}
            \begin{proof}
        	\pf\ By step \stepref{<2>3}.
            \end{proof}
	    \step{<3>2}{$\SE{\clink{l}{f},f} < \add{\T}{\eDtrans}{\clink{l}{f}, f'}$}
            \begin{proof}
        	\pf\ By step \stepref{<2>1}.
            \end{proof}
            \qedstep
            \begin{proof}
	      \pf\ By steps \stepref{<3>1} and \stepref{<3>2}, subtracting $\eDtrans{\clink{l}{f},f'}$ from both sides.
            \end{proof}
        \end{proof}
        \qedstep
        \begin{proof}
	  \pf\ Step \stepref{<2>1} and case assumption \stepref{<1>5} yield $\T{\clink{l}{f}, f'} < \T{\clink{l}{f}, f}$, contradicting step \stepref{<1>1} with $(\clink{l}{f}, f) \in \V$ by steps \stepref{<2>2} and \stepref{<2>4}.
        \end{proof}
    \end{proof}
    \step{c<2>1}{\case{$\Smin{\clink{l}{f},f} \leq \T{\clink{l}{f},f} < \SE{\clink{l}{f},f}$}}
    \begin{proof}
	\explan{
	  In case the frame $f$ starts it transmission earlier than expected, but still within $\S{\clink{l}{f},f}$, there must be some $f' \in \IE{\clink{l}{f},f}$ which is transmitted directly before $f$ (steps \estepref{2}{1} and \estepref{2}{2}).
	    We conclude that $f'$ is also transmitted earlier than expected (step \estepref{2}{3}), yielding a contradiction to step \stepref{<1>1}.
	}
	\step{<2>1}{$\IE{\clink{l}{f},f} \neq \emptyset$}
	\begin{proof}
	  \assume{$\IE{\clink{l}{f},f} = \emptyset$}
	    \prove{False}
	    \explan{Proof by contradiction.}
	    \step{<3>1}{$\SE{\clink{l}{f},f} = \Smin{\clink{l}{f},f}$}
	    \begin{proof}
	      \pf\ By definition of $\SE{}$ in (\ref{eq:fips:se}).
	    \end{proof}
	    \qedstep
	    \begin{proof}
		\pf\ Step \stepref{<3>1} contradicts case assumption \stepref{c<2>1}.
	    \end{proof}
	\end{proof}
	\step{<2>2}{\pick{ $f' \in \IE{\clink{l}{f},f}$ with $f' \in \TT$ and $\IE{\clink{l}{f},f} = \IE{\clink{l}{f},f'} \cup \{f'\}$}}
	\begin{proof}
	  \pf\ Let $\IE{\clink{l}{f},f} = \{f_1, \ldots, f_m\}$, with $m \geq 1$ by step \stepref{<2>1}, and order $f_i$ such that
	    \begin{equation*}
	      \T{\clink{l}{f}, f_1} < \T{\clink{l}{f},f_2} < \ldots < \T{\clink{l}{f}, f_m}.
	    \end{equation*}
	    Then, $f_m$ satisfies $\IE{\clink{l}{f},f} = \IE{\clink{l}{f},f_m} \cup \{f_m\}$ by definition of $\IE{}$.
	\end{proof}
	\pflet{$l'$ with $\v{l}{f} = \v{l'}{f'}$.}
	\step{<2>3}{$\T{\clink{l}{f}, f'} < \SE{\clink{l}{f},f'}$}
	\begin{proof}
	    \explan{We show the inequality step by step.}
	    \step{<3>1}{$\add{\T}{\eDtrans}{\clink{l}{f},f'} \leq \T{\clink{l}{f},f}$}
	    \begin{proof}
	      \pf\ By $\TC$, because $\T{\clink{l}{f},f'} < \T{\clink{l}{f},f}$ by step \stepref{<2>2}.
	    \end{proof}
	    \step{<3>2}{$\T{\clink{l}{f},f} < \SE{\clink{l}{f},f}$}
	    \begin{proof}
		\pf\ By case assumption \stepref{c<2>1}.
	    \end{proof}
	    \step{<3>3}{$\SE{\clink{l}{f},f} = \add{\SE}{\eDtrans}{\clink{l}{f},f'}$}
	    \begin{proof}
	      \pf\ By step \stepref{<2>2} and definition of $\SE{}$ in (\ref{eq:fips:se}).
	    \end{proof}
	    \qedstep
	    \begin{proof}
	      \pf\ By steps \stepref{<3>1}--\stepref{<3>3}, when subtracting $\eDtrans{\clink{l}{f},f'}$ from both sides.
	    \end{proof}
	\end{proof}
	\qedstep
	\begin{proof}
	  \pf\ By step \stepref{<2>2}, we have $\T{\clink{l}{f},f'} < \T{\clink{l}{f},f}$.
	    Since $\T{\clink{l}{f},f} < \SE{\clink{l}{f},f}$ by case assumption \stepref{c<2>1} and $\SE{\clink{l}{f},f} < \infty$ by definition (\ref{eq:fips:se}), we also have $\T{\clink{l}{f},f'} < \infty$.
	    Therefore, $\PSFP$ yields 
	    \begin{equation*}
	      \add{\T}{\D}{\pclink{l'}{f'},f'} \in \R{\v{l}{f}, f'}.
	    \end{equation*}
	    But then, step \stepref{<2>3} and $\T{\clink{l}{f},f'} < \T{\clink{l}{f},f}$ contradict step \stepref{<1>1}.
	\end{proof}
    \end{proof}
    \step{c<2>2}{\case{$\T{\clink{l}{f},f} < \Smin{\clink{l}{f},f}$}}
    \begin{proof}
        \explan{
	  In case the frame $f$ starts transmission earlier than expected, even before $\Smin{\clink{l}{f},f}$, we show that there must exist a GCL window $(o_j, c_j)$ during which $f$ is transmitted (step \estepref{2}{1}).
            But by construction of FIPS's GCL Encoding, there exists another frame $f'$ whose expected transmission starts at $o_j$ (step \estepref{2}{2}).
	    With C3 and the construction of the FIPS's PSFP intervals, we show in step \estepref{2}{3} that $f$ would arrive at $\v{l}{f}$ outside its PSFP enforce interval \ldash unfolding the desired contradiction.
        }
	\step{<2>1}{\pick{ $1 \leq j < i$ such that $o_j \leq \T{\clink{l}{f},f} < c_j$}}
        \begin{proof}
            \explan{
        	We start in \estepref{3}{1} to show that such $j$ exists with $1 \leq j \leq m$, before restricting it to $j < i$ in \estepref{3}{2}.
            }
	    \step{<3>1}{\pick{ $1 \leq j \leq m$ such that $o_j \leq \T{\clink{l}{f},f} < c_j$}}
            \begin{proof}
        	\pf\ By $\GCLC$.
            \end{proof}
            \step{<3>2}{$j < i$}
            \begin{proof}
        	\assume{$j \geq i$}
        	\prove{False}
        	\explan{
        	    Proof by contradiction.
		    Assume $j \geq i$ and derive that $f$ does not start its transmission earlier than $\Smin{\clink{l}{f},f}$ -- a contradiction to the case assumption \estepref{1}{6}.
        	}
		\step{<4>1}{$o_i \leq o_j \leq \T{\clink{l}{f},f}$}
        	\begin{proof}
        	    \pf\ By assumption \stepref{<3>2}, definition of FIPS's GCL Encoding, and step \stepref{<2>1}.
        	\end{proof}
        	\qedstep
        	\begin{proof}
        	    \pf\ Steps \stepref{<4>1} and \stepref{<1>4} contradict the case assumption \stepref{c<2>2}.
        	\end{proof}
            \end{proof}	
            \qedstep
            \begin{proof}
        	\pf\ By steps \stepref{<3>1} and \stepref{<3>2}.
            \end{proof}
        \end{proof}
	\step{<2>2}{\pick{ $f' \in \TT_{\clink{l}{f}}$ with $o_j = \Smin{\clink{l}{f},f'}$}}
        \begin{proof}
            \pf\ Exists by construction of FIPS's GCL Encoding.
        \end{proof}
	\step{<2>3}{$\add{\T}{\D}{\pclink{l}{f},f} \notin \R{\v{l}{f}, f}$}
        \begin{proof}
	  \explan{We show $\add{\T}{\D}{\pclink{l}{f},f} < \Rmin{\v{l}{f},f}$ step by step.}
	    \step{<3>1}{$\add{\T}{\D}{\pclink{l}{f},f} \leq \T{\clink{l}{f},f}$}
            \begin{proof}
        	\pf\ By $\ST$.
            \end{proof}
	    \step{<3>2}{$\T{\clink{l}{f},f} < c_j$}
            \begin{proof}
        	\pf\ By step \stepref{<2>1}.
            \end{proof}
	    \step{<3>3}{$c_j = \add{\Smax}{\dpdbtransmax}{\clink{l}{f},f'}$}
            \begin{proof}
        	\pf\ By step \stepref{<2>2} and by construction of FIPS's GCL Encoding.
            \end{proof}
	    \step{<3>4}{$\add{\Smax}{\dpdbtransmax}{\clink{l}{f},f'} \leq \Rmin{\v{l}{f}, f}$}
            \begin{proof}
        	\explan{We show that C3's preconditions for the strict separation are satisfied.}
		\step{<4>1}{$\R{\v{l}{f},f'} \cap \R{\v{l}{f},f} = \emptyset$}
        	\begin{proof}
		  \pf\ By C3 with $\Smin{\clink{l}{f}, f'} = o_j < o_i = \Smin{\clink{l}{f},f}$ (steps \stepref{<2>2}, \stepref{<2>1}, and \stepref{<1>4}).
        	\end{proof}
		\step{<4>2}{$\Rmin{\v{l}{f},f'} \leq \Rmax{\v{l}{f},f}$}
        	\begin{proof}
		  \assume{$\Rmin{\v{l}{f},f'} > \Rmax{\v{l}{f},f}$}
        	    \prove{False}
		    \explan{Proof by contradiction}
        	    \step{<5>1}{$o_i < o_j$}
        	    \begin{proof}
        		\pf\ With C3 and C1, we have
        		\begin{equation*}
			  o_i = \Smin{\clink{l}{f},f} < \add{\Smax}{\dpdbtransmax}{\clink{l}{f},f} \leq \Rmin{\v{l}{f},f'} \leq \Smin{\clink{l}{f},f'} = o_j
        		\end{equation*}
        	    \end{proof}
        	    \qedstep
        	    \begin{proof}
        		\pf\ Step \stepref{<5>1} contradicts step \stepref{<3>1}.
        	    \end{proof}
        	\end{proof}
		\qedstep
		\begin{proof}
		  \pf\ Steps \stepref{<4>1} and \stepref{<4>2} yield $\Rmax{\v{l}{f},f'} < \Rmin{\v{l}{f},f}$.
		    Therefore, C3 yields the inequality of \stepref{<3>4}.
		\end{proof}
            \end{proof}
            \qedstep
            \begin{proof}
	      \pf\ Steps \stepref{<3>1}--\stepref{<3>4} yield $\add{\T}{\D}{\pclink{l}{f},f} < \Rmin{\v{l}{f},f}$.
            \end{proof}
        \end{proof}
        \qedstep
        \begin{proof}
	  \pf\ Step \stepref{<2>3} contradicts $(\clink{l}{f},f) \in \V$ (step \stepref{<1>1}) with definition \toplevel:1.
        \end{proof}
    \end{proof}
    \qedstep
    \begin{proof}
	\pf\ All three cases \stepref{<1>5}, \stepref{c<2>1}, and \stepref{c<2>2} yield a contradiction.
    \end{proof}
\end{proof}

\begin{theorem}\label{theorem:fips}
  A TSN configuration $C$ derived by FIPS robustly schedules every stream $f \in \TT$.
\end{theorem}
\begin{proof}
    \textit{Proof:}
    \pflet{$f \in \TT$ frame}
    \pflet{$\v{k+1}{f} \in \path{f}$ TSN bridge}
    \pflet{$\E = (\T{}, \D{})$ valid execution scheme satisfying (5)}
    \explan{With Definition 1, we prove the statement with induction over $k$.}
    \step{<1>1}{$\D{[\v{l}{f}, \v{l+1}{f}], f} \in \dpdb{[\v{l}{f}, \v{l+1}{f}], f}, \quad \forall 1 \leq l \leq k.$}
    \begin{proof}
      \pf\ Precondition (5) of robustness definition.
    \end{proof}
    \step{<1>2}{\case{$k = 1$}}
    \begin{proof}
      \step{<2>1}{$\D{[\v{1}{f}, \v{2}{f}], f} \in \dpdb{[\v{1}{f}, \v{2}{f}], f}$}
      \begin{proof}
	\pf\ By step \stepref{<1>1} with $l = 1$.
      \end{proof}
      \step{<2>2}{$\T{[\v{1}{f}, \v{2}{f}], f} = \Smin{[\v{1}{f}, \v{2}{f}], f}$}
      \begin{proof}
	\pf\ By \IT.
      \end{proof}
      \qedstep
      \begin{proof}
	\explan{We show that $f$ arrives within its PSFP interval $[\Rmin{\v{2}{f}, f}, \Rmax{\v{2}{f}, f}]$.}
	\step{<3>1}{$\add{\T}{\D}{[\v{1}{f}, \v{2}{f}], f} \geq \Rmin{\v{2}{f}, f}$}
	\begin{proof}
	  \step{<4>1}{$\add{\T}{\D}{[\v{1}{f}, \v{2}{f}], f} = \add{\Smin}{\D}{[\v{1}{f}, \v{2}{f}], f}$}
	  \begin{proof}
	    \pf\ By Step \stepref{<2>2}.
	  \end{proof}
	  \step{<4>2}{$\add{\Smin}{\D}{[\v{1}{f}, \v{2}{f}], f} \geq \add{\Smin}{\dpdbmin}{[\v{1}{f}, \v{2}{f}], f}$}
	  \begin{proof}
	    \pf\ By Step \stepref{<1>1}.
	  \end{proof}
	  \step{<4>3}{$\add{\Smin}{\dpdbmin}{[\v{1}{f}, \v{2}{f}], f} = \Rmin{\v{2}{f}, f}$}
	  \begin{proof}
	    \pf\ By construction of the PSFP intervals in FIPS.
	  \end{proof}
	  \qedstep
	  \begin{proof}
	    \pf\ By steps \stepref{<4>1}--\stepref{<4>3}, in that order.
	  \end{proof}
	\end{proof}
	\step{<3>2}{$\add{\T}{\D}{[\v{1}{f}, \v{2}{f}], f} \leq \Rmax{\v{2}{f}, f}$}
	\begin{proof}
	  \step{<4>1}{$\add{\T}{\D}{[\v{1}{f}, \v{2}{f}], f} = \add{\Smin}{\D}{[\v{1}{f}, \v{2}{f}], f}$}
	  \begin{proof}
	    \pf\ By Step \stepref{<2>2}.
	  \end{proof}
	  \step{<4>2}{$\add{\Smin}{\D}{[\v{1}{f}, \v{2}{f}], f} \leq \add{\Smax}{\D}{[\v{1}{f}, \v{2}{f}], f}$}
	  \begin{proof}
	    \pf\ By (\ref{eq:sminmax}).
	  \end{proof}
	  \step{<4>3}{$\add{\Smax}{\D}{[\v{1}{f}, \v{2}{f}], f} \leq \add{\Smax}{\dpdbmax}{[\v{1}{f}, \v{2}{f}], f}$}
	  \begin{proof}
	    \pf\ By Step \stepref{<1>1}.
	  \end{proof}
	  \step{<4>4}{$\add{\Smax}{\dpdbmax}{[\v{1}{f}, \v{2}{f}], f} = \Rmax{\v{2}{f}, f}$}
	  \begin{proof}
	    \pf\ By construction of the PSFP intervals in FIPS.
	  \end{proof}
	  \qedstep
	  \begin{proof}
	    \pf\ By steps \stepref{<4>1}--\stepref{<4>4}, in that order.
	  \end{proof}
	\end{proof}
	\qedstep
	\begin{proof}
	  \pf\ By steps \stepref{<3>1} and \stepref{<3>2}.
	\end{proof}
      \end{proof}
    \end{proof}
    \step{<1>3}{\case{$k > 1$}}
    \begin{proof}
      \step{<2>1}{$\add{\T}{\D}{[\v{k-1}{f}, \v{k}{f}], f} \in \R{\v{k}{f}, f}$}
      \begin{proof}
	\pf\ By induction hypothesis.
      \end{proof}
      \step{<2>2}{$\T{[\v{k}{f}, \v{k+1}{f}], f} = \SE{[\v{k}{f}, \v{k+1}{f}], f}$}
      \begin{proof}
	\pf\ By Lemma~\ref{lemma:fips} with Step \stepref{<2>1}.
      \end{proof}
      \step{<2>3}{$\T{[\v{k}{f}, \v{k+1}{f}], f} \in \S{[\v{k}{f}, \v{k+1}{f}], f}$}
      \begin{proof}
	\pf\ By Lemma~\ref{lemma:fips:aux} with Step \stepref{<2>2}.
      \end{proof}
      \qedstep
      \begin{proof}
	\explan{We show that $f$ arrives within its PSFP interval $[\Rmin{\v{k+1}{f}, f}, \Rmax{\v{k+1}{f}, f}]$.}
	\step{<3>1}{$\add{\T}{\D}{[\v{k}{f}, \v{k+1}{f}], f} \geq \Rmin{\v{k+1}{f}, f}$}
	\begin{proof}
	  \step{<4>1}{$\add{\T}{\D}{[\v{k}{f}, \v{k+1}{f}], f} \geq \add{\Smin}{\D}{[\v{k}{f}, \v{k+1}{f}], f}$}
	  \begin{proof}
	    \pf\ By Step \stepref{<2>3}.
	  \end{proof}
	  \step{<4>2}{$\add{\Smin}{\D}{[\v{k}{f}, \v{k+1}{f}], f} \geq \add{\Smin}{\dpdbmin}{[\v{k}{f}, \v{k+1}{f}], f}$}
	  \begin{proof}
	    \pf\ By Step \stepref{<1>1}.
	  \end{proof}
	  \step{<4>3}{$\add{\Smin}{\dpdbmin}{[\v{k}{f}, \v{k+1}{f}], f} = \Rmin{\v{k+1}{f}, f}$}
	  \begin{proof}
	    \pf\ By construction of the PSFP intervals in FIPS.
	  \end{proof}
	  \qedstep
	  \begin{proof}
	    \pf\ By steps \stepref{<4>1}--\stepref{<4>3}, in that order.
	  \end{proof}
	\end{proof}
	\step{<3>2}{$\add{\T}{\D}{[\v{k}{f}, \v{k+1}{f}], f} \leq \Rmax{\v{k+1}{f}, f}$}
	\begin{proof}
	  \step{<4>1}{$\add{\T}{\D}{[\v{k}{f}, \v{k+1}{f}], f} \leq \add{\Smax}{\D}{[\v{k}{f}, \v{k+1}{f}], f}$}
	  \begin{proof}
	    \pf\ By Step \stepref{<2>3}.
	  \end{proof}
	  \step{<4>2}{$\add{\Smax}{\D}{[\v{k}{f}, \v{k+1}{f}], f} \geq \add{\Smax}{\dpdbmax}{[\v{k}{f}, \v{k+1}{f}], f}$}
	  \begin{proof}
	    \pf\ By Step \stepref{<1>1}.
	  \end{proof}
	  \step{<4>3}{$\add{\Smax}{\dpdbmax}{[\v{k}{f}, \v{k+1}{f}], f} = \Rmax{\v{k+1}{f}, f}$}

	  \begin{proof}
	    \pf\ By construction of the PSFP intervals in FIPS.
	  \end{proof}
	  \qedstep
	  \begin{proof}
	    \pf\ By steps \stepref{<4>1}--\stepref{<4>3}, in that order.
	  \end{proof}
	\end{proof}
	\qedstep
	\begin{proof}
	  \pf\ By steps \stepref{<3>1} and \stepref{<3>2}.
	\end{proof}
      \end{proof}
    \end{proof}
    \qedstep
    \begin{proof}
      \pf\ By induction, the statement holds for all $1 \leq k < n(f)$.
    \end{proof}
\end{proof}

\end{document}